\documentclass[11pt]{article}
\usepackage{graphicx} 
\usepackage{amsmath,amssymb,amsthm}
\usepackage{mathrsfs}
\usepackage{hyperref}
\hypersetup{
    colorlinks=true,
    linkcolor=blue,
    filecolor=magenta,      
    urlcolor=cyan,
    citecolor=blue
}
\usepackage{faktor}
\usepackage{float}
\usepackage{youngtab}
\usepackage[abs]{overpic}
\usepackage{comment}
\usepackage[a4paper,top=2.5cm,width=18cm, bottom=3.2cm]{geometry}
\usepackage{hhline}

\newcommand*{\Resize}[2]{\resizebox{#1}{!}{$#2$}}%

\usepackage{wrapfig}
\usepackage{tikz-cd}
\usetikzlibrary{braids} 
\usepackage{arydshln}
\usepackage[skins,theorems]{tcolorbox}
\usepackage{fdsymbol}
\usepackage{lscape}
\usepackage{float}

\usepackage{comment}
\usepackage[font=small,labelfont=bf]{caption}
\usepackage{multicol}
\usepackage{multirow}
\usepackage{booktabs}

\definecolor{amber}{rgb}{1.0, 0.49, 0.0}
\definecolor{Green}{rgb}{0.0, 0.5, 0.0}
\definecolor{purple}{rgb}{0.7,0,0.7}

\newcommand{\pk}[1]{\textbf{\color{purple} [{\sc pk}: {#1}]}}

\theoremstyle{plain}
\newtheorem{thm}{\protect\theoremname}[section]
\theoremstyle{plain}
\newtheorem{conjecture}[thm]{\protect\conjecturename}
\theoremstyle{definition}
\newtheorem{problem}[thm]{\protect\problemname}
\theoremstyle{definition}
\newtheorem{defn}[thm]{\protect\definitionname}
\theoremstyle{definition}
\newtheorem*{defn*}{\protect\definitionname}
\theoremstyle{plain}
\newtheorem*{thm*}{\protect\theoremname}
\theoremstyle{definition}
\newtheorem*{sol*}{\protect\solutionname}
\theoremstyle{definition}
\newtheorem*{example*}{\protect\examplename}
\theoremstyle{plain}

\theoremstyle{definition}
\newtheorem*{problem*}{\protect\problemname}
\theoremstyle{plain}
\newtheorem*{conjecture*}{\protect\conjecturename}
\theoremstyle{plain}
\newtheorem{rmk}[thm]{Remark}
\theoremstyle{plain}
\newtheorem{proposition}[thm]{Proposition}

\makeatother

\providecommand{\conjecturename}{Conjecture}
\providecommand{\definitionname}{Definition}
\providecommand{\examplename}{Example}
\providecommand{\factname}{Fact}
\providecommand{\problemname}{Problem}
\providecommand{\solutionname}{Solution}
\providecommand{\theoremname}{Theorem}

\newcommand{\overcross}{
 {\mathchoice
  {\includegraphics[height=1.6ex]{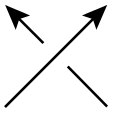}}
  {\includegraphics[height=1.6ex]{figures/overcrossing.png}}
  {\includegraphics[height=1.2ex]{figures/overcrossing.png}}
  {\includegraphics[height=0.9ex]{figures/overcrossing.png}}
 }
}
\newcommand{\undercross}{
 {\mathchoice
  {\includegraphics[height=1.6ex]{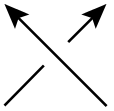}}
  {\includegraphics[height=1.6ex]{figures/undercrossing.png}}
  {\includegraphics[height=1.2ex]{figures/undercrossing.png}}
  {\includegraphics[height=0.9ex]{figures/undercrossing.png}}
 }
}
\newcommand{\nocross}{
 {\mathchoice
  {\includegraphics[height=1.6ex]{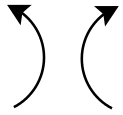}}
  {\includegraphics[height=1.6ex]{figures/no_crossing.png}}
  {\includegraphics[height=1.2ex]{figures/no_crossing.png}}
  {\includegraphics[height=0.9ex]{figures/no_crossing.png}}
 }
}
\usepackage{authblk} 
\usepackage{orcidlink}

\title{Full twists and stability of knots and quivers}

\date{\phantom{2025}}

\author[1,2,3]{Sachin Chauhan \orcidlink{0000-0003-1045-3428}\thanks{sachin.chauhan@math.uu.se}}
\author[4]{Piotr Kucharski \orcidlink{0000-0002-9599-5658}\thanks{piotr.kucharski@mimuw.edu.pl}}
\author[5]{Dmitry Noshchenko \orcidlink{0000-0002-9639-5603}\thanks{dsnoshchenko@stp.dias.ie}}
\author[6]{\break Ramadevi Pichai\orcidlink{0000-0003-4331-8222}\thanks{ramadevi@iitb.ac.in}}
\author[7]{Vivek Kumar Singh \orcidlink{0000-0001-9141-2331}\thanks{vks2024@nyu.edu} }
\author[8,9]{Marko Sto\v
si\'c\orcidlink{0000-0002-4464-396X}\thanks{mstosic@fc.ul.pt}}

\affil[1]{Centre for Geometry and Physics, Uppsala University, Box 516, 751 20 Uppsala, Sweden}
\affil[2]{ Department of Mathematics, Uppsala University, Box 480, 751 06 Uppsala, Sweden}
\affil[3]{Department of Physics and Astronomy, Uppsala University, Box 516, 751 20 Uppsala, Sweden}
\affil[4]{Institute of Mathematics, University of Warsaw, ul. Banacha 2, 02-097 Warsaw, Poland}
\affil[5]{School of Theoretical Physics, Dublin Institute for Advanced Studies,\break 10 Burlington Road, Dublin 4, D04 C932, Ireland}
\affil[6]{Department of Physics, Indian Institute of Technology Bombay,\break Powai, Mumbai 400076, India}
\affil[7]{Center for Quantum and Topological Systems (CQTS), NYUAD Research Institute,\break
New York University Abu Dhabi, PO Box 129188, Abu Dhabi, UAE}
\affil[8]{CEMS.UL, Departamento de Matem\'atica, Faculdade de Ci\^encias, Universidade de Lisboa,\break Edif\'icio C6, Campo Grande, 1749-016 Lisboa, Portugal}
\affil[9]{Mathematical Institute SANU, Knez Mihajlova 36, 11000 Belgrade, Serbia}

\begin{document}

\hfill
\begin{tabular}{r}
UUITP-24/25\\
DIAS-STP-25-21
\end{tabular}
{\let\newpage\relax\maketitle}
\maketitle

\abstract{We relate the stability of knot invariants under twisting a pair of strands to the stability of symmetric quivers under unlinking (or linking) operation. Starting from the HOMFLY-PT skein relations, we confirm the stable growth of $Sym^r$-coloured HOMFLY-PT polynomials under the addition of a~full twist to the knot. On the other hand, we show that symmetric quivers exhibit analogous stable growth under unlinking or linking of the quiver augmented with the extra node; in some cases this augmented quiver captures the spectrum of motivic Donaldson-Thomas invariants of all quivers in the sequence. Combining these two versions of the stable growth, we conjecture that performing a~full twist on any knot corresponds to appropriate unlinking or linking of the corresponding augmented quiver -- this statement is an important step towards a~direct definition of the knot-quiver correspondence based on the knot diagram. We confirm the conjecture for all twist knots, $(2,2p+1)$ torus knots, and all pretzel knots up to 15 crossings with an~odd number of twists in each twist region.}
\newpage
\tableofcontents
\newpage
%
%
%
%
%
%

\section{Introduction and summary}

One of the most efficient strategies in quantum topology is to focus on understanding how invariants of links and low-dimensional manifolds behave under various operations which change their structure. Among such, a~simple transformation of links in $S^3$ is an addition of a~{\it full twist} on a~pair of strands. Under this operation, HOMFLY-PT homologies of torus links exhibit stability properties \cite{stosic2007homological,gorsky2013stable,gorsky2014torus,gorsky2015stable}, which implies the existence of a~well-defined limit when the number of twists goes to infinity.
Analogous behaviour of the coloured Jones polynomials for a~general link was considered in \cite{lee2019stability,lee2022colored,lee2025stable}. Different, but related forms of stability for quantum link invariants were discussed in \cite{garoufalidis2010degree,garoufalidis2013stability,garoufalidis2015nahm}. 

Another fruitful approach to quantum topology emerges from string theory dualities exemplified by the {\it knot-quiver correspondence} \cite{KRSS1707short,KRSS1707long}. The correspondence is motivated by studying the supersymmetric quiver quantum mechanics description of BPS states in brane systems describing knots \cite{OV9912, Rei12, KS1608}.
Specifically, the symmetrically coloured HOMFLY-PT polynomials of a~knot $K$ were related to the motivic generating series of a~symmetric quiver, denoted $Q_K$. Subsequent studies showed that symmetric quivers also have topologically relevant operations called {\it unlinking}, {\it linking} \cite{EKL1910}, and {\it splitting} \cite{JKLNS2105,KLNS2312}.

In this paper, we combine strategies mentioned above to relate the stability of knot invariants under twisting to the stability of symmetric quivers under unlinking (or linking) operation. It allows us to make a~substantial step towards one of the most important questions about the knot-quiver correspondence:
\begin{problem*}
    How to directly connect knot diagrams with quivers?
\end{problem*}
Our (partial) answer takes the form of the knot-quiver stable growth conjecture, which in simplified manner can be stated as follows (proper formulation is presented in Conjecture~\ref{coj:knot quivers twists}): 

\begin{conjecture*}
    Twisting strands in the knot diagram corresponds to unlinking (or linking) the quiver.
\end{conjecture*}

In order to make this more precise, we introduce a~new kind of symmetric quiver for a~knot $K$ which we call \emph{augmented quiver} and denote $Q_K^+$. It contains $Q_K$ as a~subquiver, and has one extra node with some number of loops and arrows to the other nodes of $Q_K$. Unlinking (or linking) these connecting arrows will produce new quivers corresponding to knots which differ from $K$ by adding a~full twist in some twist region, which is specified by a~proper choice of augmentation arrows.
For example, Figure \ref{fig:augmentation_unlinking} shows one possible augmentation of $Q_{4_1}$ with the extra node and arrows shown in red. Unlinking those arrows yields a~new quiver, which matches (after removing the red node) a~quiver $Q_{6_1}$! (See Section \ref{sec:41,61,81} for this concrete example.)
\begin{figure}[ht!]
    \centering
    \includegraphics[scale=0.45]{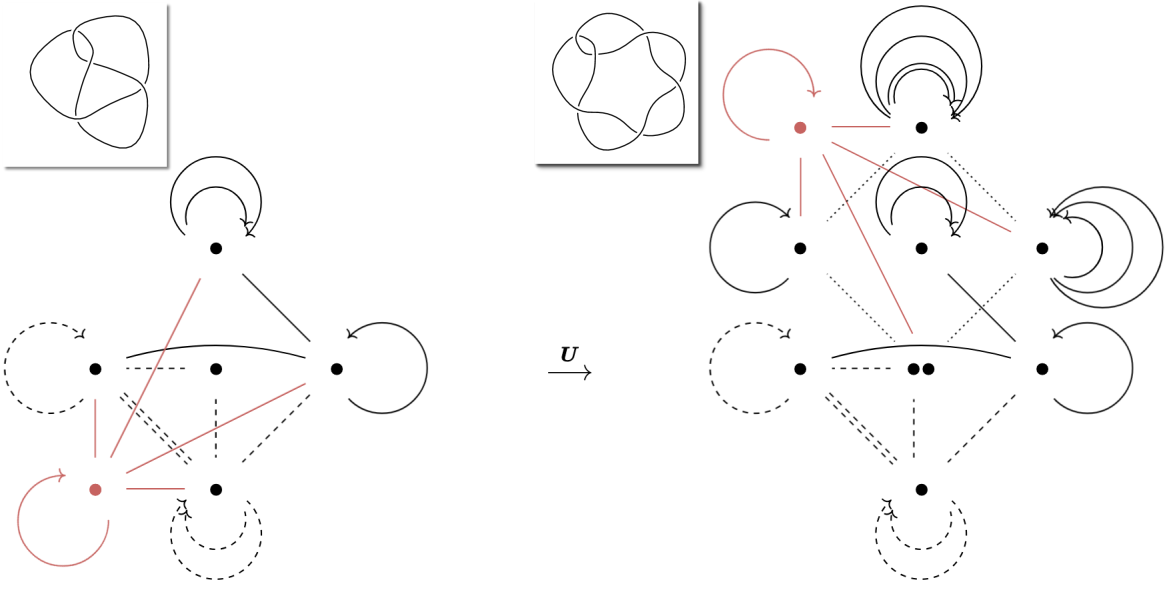}
    \captionsetup{width=0.8\linewidth}
    \caption{Left: augmented quiver $Q_{4_1}^+$ for the figure-eight knot (here negative arrows are shown with dashed lines, and a~single line between two nodes corresponds to a~pair of arrows in both directions). Augmentation node and arrows are shown in red. Right: unlinking the four pairs of arrows connecting the extra node produces a~quiver $Q_{6_1}$, after removal of that node (for simplicity, some of the arrows in $Q_{6_1}$ are omitted).
    This demonstrates how unlinking the quiver can be seen as twisting the corresponding knot, transforming $4_1$ into $6_1$.
    }
    \label{fig:augmentation_unlinking}
\end{figure}

This may seem like a~pure magic, but the foundation of the above conjecture lies in the studies of stable growth of knots and quivers and matching the structure on both sides. The organisation of the paper reflects this idea:

\begin{itemize}
\item Section \ref{sec:Stable growth for knots} focuses on knots and links. We start from the HOMFLY-PT skein relation and move on to generalised skein relations for $r$-coloured HOMFLY-PT $P_r(K;a,q)$ for a~sequence of knots $\{K_i\}$ \cite{RGK}. Particularly, we focus on a~twist region of knot $K_i$ (see Figure \ref{fig:full_twist}) and add full twists to generate the sequence of knots in the skein relation. We discuss the stable growth of coloured HOMFLY-PT polynomials of knots ($K_{\infty}$) when the full twists between antiparallel strands tends to infinity. Also,  we present the stable growth of $P_r(K_{\infty};a,q)$  when infinite full-twists are between parallelly oriented strands. Explicit form for $P_r(T(2p+1);a,q)$ (up to a~suitable normalisation) for the torus knots $T(2,2p+1)$ when the full twists $p \rightarrow \infty$ is  proven to be proportional to the unknot ($\bigcirc$) polynomial $P_{2r}(\bigcirc;a,q)$.

\item In Section \ref{sec:Quiver perspective} we study the stable growth of quivers. We discuss the operation of splitting, introduced in \cite{JKLNS2105} to construct different quivers corresponding to the same knot, as well as unlinking and linking, defined in \cite{EKL1811} basing on the skein relations for boundaries of holomorphic disks \cite{ES1901}. We also discuss their relation, showing how splitting can be engineered by unlinking and linking applied to the quiver with the extra node. 
Then, we recursively use these operations to construct sequences of quivers $Q_1 \subset Q_2  \subset \ldots \subset Q_{\infty}$ and analyse their stability. In some special cases, we are able to capture the spectrum of motivic Donaldson-Thomas invariants \cite{KS0811,KS1608} of all quivers in the sequence using a~single quiver with the extra node (see Theorem \ref{thm:BPS_spectrum_of_augmented_quiver}). 

\item In Section \ref{sec:knot-quiver_growth} we compare the results from two previous sections: the statement that twisting leads to the stable growth of the knot with the fact that unlinking (or linking) leads to the stable growth of the~quiver. Most importantly, we combine those two stable behaviours into one, as stated in Conjecture~\ref{coj:knot quivers twists}. We also define and then interpret augmented quivers for knots in terms of link invariants. This leads to many interesting consequences, for example, it suggests that knot homologies of a~series of knots which differ by a~full twist can be seen as a~part of a~larger homology theory enriched with differentials corresponding to twisting (although we leave these ideas for future research, and in this paper focus on describing the corresponding symmetric quivers).

\item In Section \ref{sec:Case studies} we work out explicitly symmetric quivers corresponding to knots which differ by a~full twist, using the unlinking (or linking) operation on a~quiver  with an extra node. We elaborate our procedure for some families of twist, torus, as well as more general pretzel knots. Further, we also present similar unlinking procedure to obtain quivers associated to $T(2,2p+1)$ torus knot complements.
\end{itemize}

\section{Stable growth for knots}\label{sec:Stable growth for knots}
In this section, we introduce families of knots obtained by adding full twists in a~twist region and study the behaviour of coloured HOMFLY-PT polynomials. We further show that, for anti-parallel strand orientations, these polynomials satisfy explicit recursion relations and converge to a~well-defined limit, while the parallel case is more subtle and lacks a~general limit formula.

\subsection{Prerequisites}

Let $L$ be a~link and $D_L$ its planar diagram (we always consider $D_L$ a~reduced diagram with respect to Reidemeister moves). 

\begin{defn}
    A twist region of $L$ is formed by maximal collections of bigon regions in $D_L$ arranged end to end. Note that a~single crossing adjacent to no bigons is \emph{not} considered as a~twist region.
\end{defn}

\noindent For example, twist regions of $8_{13}$ knot are shown in Figure \ref{fig:8_13_twist} -- there is one twist region with two full twists and the other one with only one full twist.

\begin{figure}[ht!]
    \centering
    \includegraphics[width=0.25\linewidth]{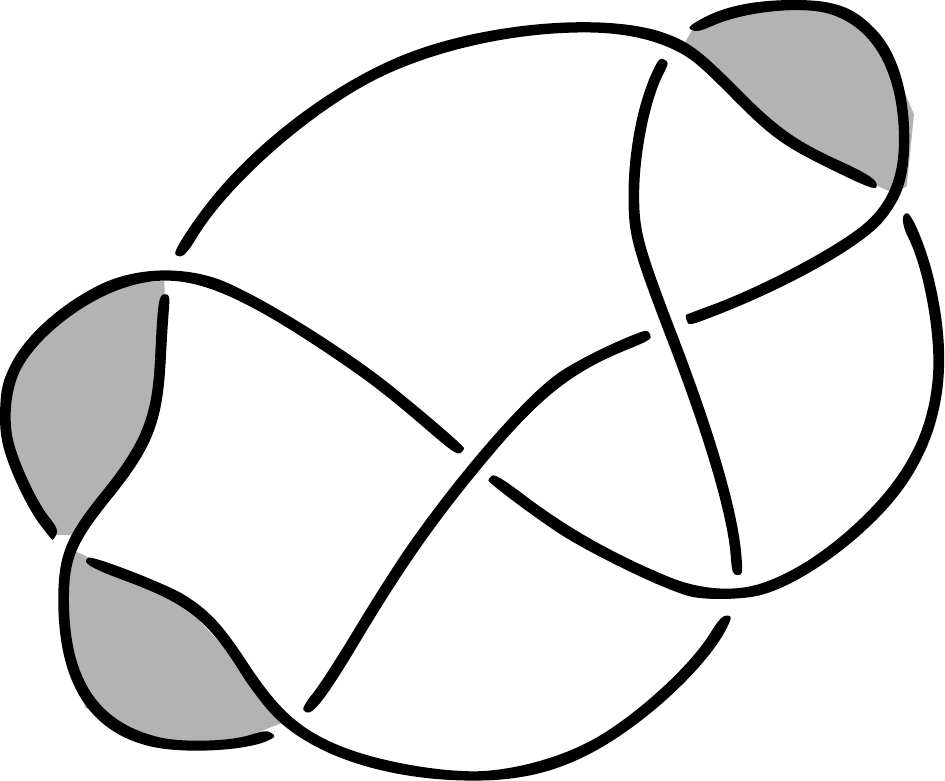}
    \caption{$8_{13}$ knot has two twist regions, shown as grey shaded areas.}
    \label{fig:8_13_twist}
\end{figure}

Let now $K$ be a~knot with twist region $\tau$ consisting of some number of connected bigons. Furthermore, assume that removing a~full twist from $\tau$ resolves this twist region, or produces an unknot.
\begin{defn}\label{defn:twisted family}
  We say that a~family of knots $\{K_i\}_{i\in\mathbb{Z}_+}^\tau$ is generated by $K_1$, if $K_{i+1}$ comes from $K_1$ by adding $i$ full twists to $K_1$ in $\tau$ (Figure \ref{fig:full_twist}). We also denote $K_0$ a~knot obtained from removing a~full twist from $K_1$ in $\tau$, as well as $K_{\infty}$ a~knot with an infinite number of twists.
\end{defn}
One of the simplest examples are twist knots, obtained from adding a~number of full twists to $K_1=4_1$: $K_2=6_1$, $K_3=8_1$, and so on (in this case $K_0=\bigcirc$).
Note, however, that many other knots can be related by this operation -- for example, adding a~full twist to $8_{13}$ produces $10_{10}$ knot (in this case $K_0=6_3$), etc.

\begin{figure}[ht!]
    \centering
    \includegraphics[width=0.5\linewidth]{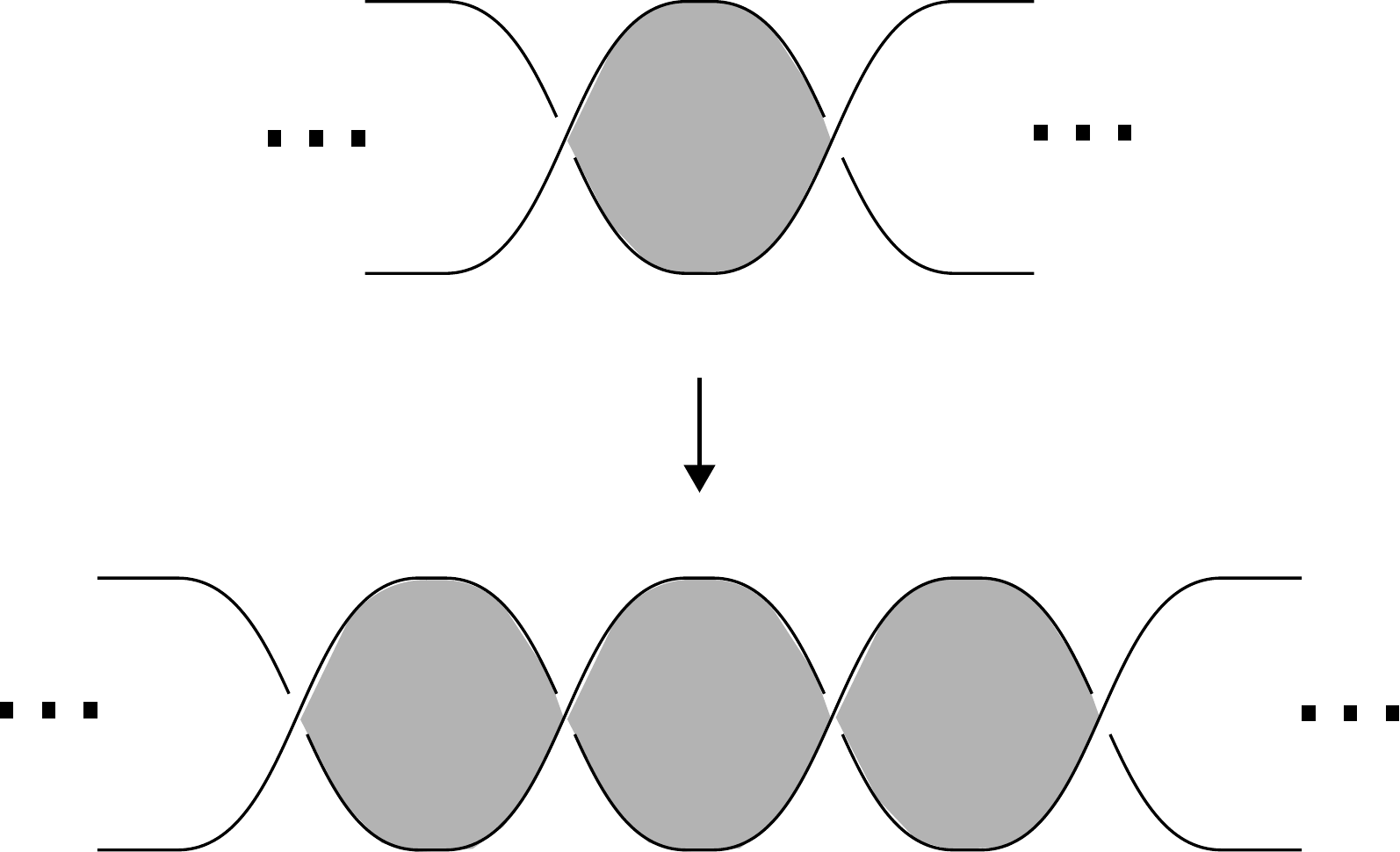}
    \captionsetup{width=0.8\linewidth}
    \caption{Adding a~full twist in a~twist region (shown in gray) is a~local operation which modifies the link $L$ by twisting a~pair of parallel strands \emph{twice}. Note that the two strands can either have the same (parallel) or the opposite (anti-parallel) orientation. The two cases will exhibit a~different behaviour of invariants.
    }
    \label{fig:full_twist}
\end{figure}

In this Section we shall exploit the behaviour of coloured HOMFLY-PT polynomials  for knots that are related by adding full twists at the twist region, akin to Figure \ref{fig:full_twist}. In particular we shall also be interested in what happens when the number of added twists goes to infinity. The existence of such limits, the so-called stability property, has been well known, especially in the case of torus knots (even with arbitrary number of strands).
The stability property seems to be one of the fundamental properties of HOMFLY-PT invariants since it even holds naturally in the case of homology, see for example \cite{stosic2007homological,turner2008spectral,rozansky2014stable,etingof2015representations,gorsky2018quadruply}. Briefly, when the number of added twists grow, the homology stabilises (up to certain homological degree, which also grows with the number of twists).

\paragraph{} Throughout the paper we use the reduced $Sym^r$-coloured symmetric representations and denote the corresponding HOMFLY-PT invariant $P_r(L;a,q)$ (we sometimes drop $a,q$ for brevity) with the unknot normalisation $P_r(\bigcirc) = 1$. For example, for the left-handed trefoil knot we have
\[
P_1(3_1) = a^2q^{-2} + a^2q^2 - a^4\,. 
\]
Link invariant $P_1(L;a,q)$ for two-component link $L$ can be brought into polynomial form after multiplying by $(q-q^{-1})$ times a~suitable monomial: we denote this polynomial as $P^{\text{fin}}(L_{\infty})$, and similarly for higher colours. Note that in this paper, we mostly focus on knots, but two-component links appear inevitably when resolving a twist region, and we have to take them into account as well. For example, for Hopf link, numbered $\mathrm{L2a1}$ in the Thistlethwaite link table, such polynomial takes form
\[
P^{\text{fin}}(\mathrm{L2a1}) = a(q-q^{-1})P_1(\mathrm{L2a1}) = 1+a^{-2}-q^{-2}-q^2\,,
\]
and for its mirror image
\[
P^{\text{fin}}(m\mathrm{L2a1}) = a^{-1}(q^{-1}-q)P_1(m\mathrm{L2a1}) = 1+a^{2}-q^{2}-q^{-2}\,.
\]
We also write $P(L;a,q)=P_1(L;a,q)$.

\paragraph{} The usual (uncoloured) HOMFLY-PT skein relation is especially important in our study:
\begin{equation}\label{eq:homfly-pt skein}
aP_1(\overcross)-a^{-1}P_1(\undercross)=(q-q^{-1})P_1(\nocross)\, .
\end{equation}
The skein relations extend to coloured HOMFLY-PT polynomials too, via the introduction of trivalent graphs and their evaluations. They are introduced for antisymmetric colours, via MOY (Murakami-Ohtsuki-Yamada) calculus \cite{MOY98} for the one-variable version ($sl(n)$ specializations), and the straightforward extensions for the two-variable antisymmetrically coloured polynomials can be found in \cite{SW1711,CKM14,TVW17}. The case of symmetric coloured HOMFLY-PT polynomials follows from the well-known symmetry:
\begin{equation}
    P_{\Lambda^r}(K;a,q)=P_{Sym^r}(K;a,q^{-1}),
\end{equation}
see e.g. \cite{lin2010homfly} (also called ``mirror symmetry" in \cite{GS,gorsky2018quadruply}).

\paragraph{} There are also homological versions of the coloured HOMFLY-PT polynomials that categorify them. However, they in general hard to compute, yet many structural properties should exist in these homological theories \cite{GS,gorsky2018quadruply}.
Although the knots-quivers  correspondence contains the information about the homological structure too (see e.g. \cite{KRSS1707short,KRSS1707long}), the homologies will not be too relevant for this paper.

\paragraph{} Going back to polynomials, for the uncoloured case we can use skein relations to directly obtain the relationship between
$P_1(K_{i+2})$, $P_1(K_{i+1})$ and $P_1(K_i)$. For higher colours, we can ``trade-off" the appearance of different trivalent resolutions (see Figure 1 of \cite{SW1711}, or \cite{MOY98}), to obtain an explicit relationship between $P_r(K_{i})$, $P_r(K_{i+1})$, $\ldots$, $P_r(K_{i+r+1})$. This has also been obtained by using Chern-Simons theory~\cite{RGK}, where the relationship takes the form:
\begin{equation}\label{gen1}
\sum_{i=0}^{r+1} (-1)^i \epsilon_i(\lambda_0,\lambda_1,\ldots,\lambda_{r}) P_r(K_{j+r+1-i}) =0,\quad j\ge 0,
\end{equation}
 for certain numbers $\lambda_{\ell}$ which are independent of knot $K$ and indices $i,j$. In the formula above $\epsilon_i$ denotes the $i$-th elementary symmetric polynomial. Although the form of the relation is the same, the exact values of the so-called ``eigenvalues" $\lambda_i$ differ in the case of parallel and anti-parallel orientations. We shall use these explicit forms below in (\ref{genra}) and (\ref{genrp}).

%

\subsection{Twisting and recursions for HOMFLY-PT polynomials}

In this Section we analyse how skein relations can be translated to ``conservation relations" for HOMFLY-PT polynomials of knots obtained by applying full twists.

\subsubsection{Antiparallel orientations}
First we shall focus on antiparallel orientations, when the two strands involved in the twist region have antiparallel (i.e. opposite) orientations. The simplest examples of knots obtained by adding full twists in this way are twist knots.

\paragraph{General recursions.}
For arbitrary colour $r$, for antiparallel orientations, we have the following: 
\begin{equation}\label{genra}
\sum_{i=0}^{r+1} (-1)^i \epsilon_i(\lambda_0,\lambda_1,\ldots,\lambda_{r}) P_r(K_{j+r+1-i}) =0,\quad j\ge 0,
\end{equation}
where $\lambda_{\ell}=a^{2\ell} q^{2\ell(\ell-1)}$, $\ell=0,\ldots,r$, and $\epsilon_i$ is the $i$-th elementary symmetric polynomial.

We can rewrite equivalently in the following ``conservation relation" form:
\begin{equation}\label{eq:homflypt twist recursion anti-parallel with const}
\sum_{i=0}^{r} (-1)^i \epsilon_i(\lambda_1,\ldots,\lambda_{r}) P_r(K_{j+r-i}) = \lambda_0^j \,\, \text{const}=\text{const},
\end{equation}
which is independent of $j\ge 0$. In addition, in this case the right-hand side (i.e. const) can be obtained from the skein relation too:
\begin{defn}
If $K$ is a~knot (and consequently all $K_j$ are knots), then the constant in (\ref{eq:homflypt twist recursion anti-parallel with const}) corresponds to the two component link, denoted by $L_{\infty}$ which is obtained by performing the 1-resolution in the crossings of the two strands that have been twisted (Figure \ref{fig:twist_to_link}).
\end{defn}
\begin{figure}[ht!]
    \centering
    \includegraphics[width=0.5\linewidth]{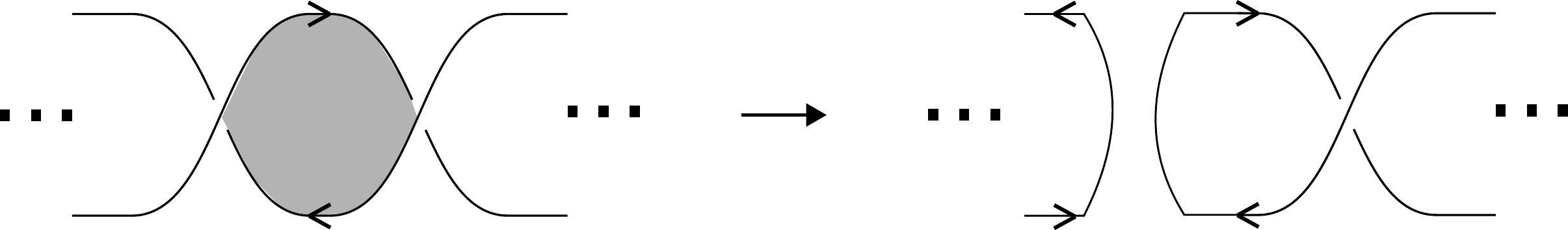}
    \captionsetup{width=0.8\linewidth}
    \caption{Modification of one crossing in the twist region which creates a~link in the case of anti-parallel orientation of strands.}
    \label{fig:twist_to_link}
\end{figure}
Then:
\begin{equation}\label{eq:homflypt twist recursion anti-parallel}
\sum_{i=0}^{r} (-1)^i \epsilon_i(\lambda_1,\ldots,\lambda_{r}) P_r(K_{j+r-i}) = [2]_{a q}[3]_{a q^2}\cdots [r]_{a q^{r-1}} P^{\text{fin}}_r (L_{\infty}),
\end{equation}
where $P^{\text{fin}}_r(L_{\infty})$ is the finite-dimensional version of the reduced HOMFLY-PT (\cite{Spectra16}) of the link $L_{\infty}$, and quantum numbers $[m]_q$ are defined for positive integers $m$ as
$$[m]_q=\frac{1-q^{2m}}{1-q^2}.$$ 

\paragraph{Colours 1 and 2.}
Colours 1 and 2, that is uncoloured HOMFLY-PT polynomial and second symmetric HOMFLY-PT polynomial, are of particular interest for the knot-quiver correspondence. In these particular cases the general formulas for arbitrary $r$ simplify as follows.
\begin{itemize}
    \item For the uncoloured polynomial, we have
\begin{equation}
P_1(K_{j+2})-(1+a^2)P_1(K_{j+1})+a^2 P_1(K_j)=0,\quad j\ge 0,
\end{equation}
as well as
\begin{equation}\label{f1}
P_1(K_{j+1})-a^2P_1(K_j)= P^{\text{fin}}_1(L_{\infty}), \quad j\ge 0.
\end{equation}
    \item In the case of second coloured HOMFLY-PT polynomial, we have:
\begin{equation}
P_2(K_{j+3})-(1+a^2+a^4q^4)P_2(K_{j+2})+(a^2+a^4q^4+a^6q^4) P_2(K_{j+1})-a^6q^4 P_2(K_j)=0,\quad j\ge 0,
\end{equation}
as well as
\begin{equation}
P_2(K_{j+2})-(a^2+a^4q^4)P_2(K_{j+1})+a^6q^4 P_2(K_j)= (1+a^2q^2) P^{\text{fin}}_2(L_{\infty}), \quad j\ge 0.
\end{equation}
\end{itemize}

For future purposes let us rewrite the last equation in the equivalent form:
\begin{align}\left(P_2(K_{j+2})-a^4q^4 P_2 (K_{j+1}) - P^{\text{fin}}_2(L_{\infty})\right)&-a^2\left(P_2(K_{j+1})-a^4q^4 P_2 (K_{j}) - P^{\text{fin}}_2(L_{\infty})\right)=\label{f2}\\
&\qquad =(1+q^2)a^2 P^{\text{fin}}_2(L_{\infty}).\nonumber
\end{align} 
\begin{rmk}
Let us take $q=1$ specialization and denote the corresponding specialization of $P_r$ by $p_r =p_r(a)=P_r(a,q=1)$. Then we have $p_2 =p_1^2$, and we can combine recursion relations (\ref{f1}) and (\ref{f2}). Indeed, from (\ref{f1}) we have:
$$p_2(K_{j+1})=(p_1(K_{j+1}))^2=(a^2 p_1(K_j)+p_1^{\text{fin}}(L_{\infty}))^2=a^4p_1^2(K_j)+(p^{\text{fin}}_1(L_{\infty}))^2+2a^2p_1(K_j)p_1^{\text{fin}}(L_{\infty}),$$
and so
$$p_2(K_{j+1})-a^4p_2(K_j)-p^{\text{fin}}_2(L_{\infty})=2a^2p_1(K_j)p_1^{\text{fin}}(L_{\infty}).$$
Then the $q=1$ specialization of the LHS of (\ref{f2}) becomes:
$$2a^2 p_1^{\text{fin}}(L_{\infty}) (p_1(K_{j+1})-a^2 p_1(K_j))=2a^2 p_1^{\text{fin}}(L_{\infty}) p_1^{\text{fin}}(L_{\infty})=2a^2 p_2^{\text{fin}}(L_{\infty}),$$
in accordance with the RHS of (\ref{f2}).
\end{rmk}

\subsubsection{Parallel orientation}

In the case when two strands in the twist region  have parallel orientations, the knots with different number of full twists also satisfy the recursion relation  of the same form (\ref{gen1}). 
The main difference is that the coefficients (``eigenvalues" $\lambda_i$) in the recursions are different than in the case of antiparallel orientation of two strands.

\paragraph{General recursion for parallel strands twists.}

For arbitrary colour $r$
we have 
\begin{equation}\label{genrp}
\sum_{i=0}^{r+1} (-1)^i \epsilon_i(\lambda_0,\lambda_1,\ldots,\lambda_{r}) P_r(K_{j+r+1-i}) =0,\quad j\ge 0,
\end{equation}
where $\lambda_{\ell}=  a^{2r} q^{2r^2-2\ell(\ell+1)}  $, $\ell=0,\ldots,r$, and $\epsilon_i$ is the $i$-th elementary symmetric polynomial.

\paragraph{Uncoloured HOMFLY-PT polynomial.}

For the uncoloured HOMFLY-PT polynomial in the case of twists in strands with parallel orientation the formula (\ref{genrp}) becomes:

\begin{equation}
P_1(K_{i+2})-a^2q^{-2}(1+q^4)P_1(K_{i+1}) +a^4 P_1(K_{i})=0.
\end{equation}

There are two different ``conservation relation" forms that one can obtain from this:
\begin{equation}
\begin{split}
    P_1(K_{j+1})-a^2q^{-2}P_1(K_{j}) &=(a^2q^2)^j \times \text{const}, \\
    P_1(K_{j+1})-a^2q^2 P_1(K_{j}) &= (a^2q^{-2})^j\times \text{const}_2.
\end{split}
\end{equation}

\paragraph{Second symmetric colour.}
For the second symmetric coloured HOMFLY-PT polynomial in the case of twists in strands with parallel   orientations, the corresponding recursions relation (\ref{genrp}) for $r=2$ can be written as: 
\begin{equation}
\begin{aligned}
a^{-4}q^4P_2(K_{i+3})-&(1+q^{12})P_2(K_{i+2}) +a^4 q^8 P_2(K_{i+1})= \\ & a^4q^4 (a^{-4}q^4 P_2(K_{i+2})-(1+ q^{12})P_2(K_{i+1}) +a^4 q^8 P_2(K_{i})).
\end{aligned}
\end{equation}
Therefore
\begin{equation}
a^{-4}q^4P_2(K_{j+2})-(1+ q^{12})P_2(K_{j+1}) +a^4 q^8 P_2(K_{j})=(a^4q^4)^j \times \text{const},
\end{equation}
where
$$\text{const}= a^{-4}q^4P_2(K_{2})-(1+ q^{12})P_2(K_{1}) +a^4 q^8 P_2(K_{0}).$$

\subsection{Limits for infinite twists}\label{sec:Limits for infinite twists}

Let $K_1,\dots,K_{\infty}$ be a~sequence of twisted knots (or links) where the strands in the corresponding twist region $\tau$ are anti-parallel (e.g. twist knots), and let
$L_{\infty}$ denote a~link obtained from resolving $\tau$. Before we argued that $\{P_r(K_i)\}$ has a~stable limit as $i\to\infty$, and we denote this limit as $P_r(K_{\infty})$. It is best to think of $P_r(K_{\infty})$ as a~rational function of $a,q$, obtained from WZNW conformal block (we may also think of some other interpretations).
\begin{proposition}\label{prp:infinite tw knot into a link}
$P_r(K_{\infty})$ satisfies the following relation:
\begin{equation}\label{eq:infinite tw knot into a link}
    P_r(K_{\infty}) = \frac{P_r(L^{\text{fin}}_{\infty})}{(a^2;q^2)_r} \simeq \frac{P_r(L_{\infty})}{\overline{P}_r(0_1)}.
\end{equation}
\end{proposition}
\begin{proof}
For general $r$, the twist recursion (\ref{eq:homflypt twist recursion anti-parallel}) takes form
\begin{equation}\label{eq:general twist recursion}
    \sum_{i=0}^r (-1)^i \epsilon_i(\lambda_1,\dots,\lambda_r)P_r(K_{j+r-i}) = [2]_{aq}[3]_{aq^2}\dots [r]_{aq^{r-1}}P^{\text{fin}}_r(L_{\infty}),
\end{equation}
where $\epsilon_i$ is the $i$-th elementary symmetric polynomial, $[n]_q:=\frac{1-q^{2n}}{1-q^2}$, and $\lambda_l = a^{2l}q^{2l(l-1)}$ for $l=0,\dots,r$ is the braiding eigenvalue.
Recall that $\epsilon_i$ appear in a~linear factorization of a~monic polynomial:
\begin{equation}
    \prod_{j=1}^r(t-\lambda_j) = t^r - \epsilon_1(\lambda_1,\dots,\lambda_r)t^{r-1} + \epsilon_2(\lambda_1,\dots,\lambda_r)t^{r-2} + \dots + (-1)^r \epsilon_r(\lambda_1,\dots,\lambda_r).
\end{equation}
Taking $j\to\infty$ in (\ref{eq:general twist recursion}) while keeping $r$ fixed, gives
\begin{equation}
    \prod_{j=1}^r(1-\lambda_j)P_r(K_{\infty}) = [2]_{aq}[3]_{aq^2}\dots [r]_{aq^{r-1}}P^{\text{fin}}_r(L_{\infty}).
\end{equation}
Plugging the exact eigenvalues, dividing both sides by $(1-a^2)$ and noting that $\lambda_1=a^2$, we can rewrite it as
\begin{equation}
\begin{aligned}
    \prod_{j=2}^r(1-\lambda_j)P_r(K_{\infty}) =  &\
    \frac{(1-a^4q^4)(1-a^6q^{12})\dots(1-a^{2r}q^{2r(r-1)})}{(1-a^2)(1-a^2q^2)(1-a^2q^4)\dots (1-a^2q^{2(r-1)})}P_r^{\text{fin}}(L_{\infty}) \\
    = &\ 
    \frac{\prod_{j=2}^r(1-\lambda_j)}{(a^2;q^2)_r} P_r^{\text{fin}}(L_{\infty}),
\end{aligned}
\end{equation}
where the two prefactors with $\lambda_i$ now cancel each other, which confirms (\ref{eq:infinite tw knot into a link}).  
\end{proof}
We stress that both sides of (\ref{eq:infinite tw knot into a link}) are rational functions in $a$ and $q$: $P_r(K_{\infty})$ can be thought of as the limit of the corresponding WZNW expression with $r\to\infty$, while the other side is the finite-dimensional link invariant divided by the unknot factor.
Multiplying both sides twice by the latter, we can also write the ``homogeneous'' version
\begin{equation}\label{eq:phopf infinite}
    \overline{P}_r(0_1)\overline{P}_r(K_{\infty}) = \overline{P}_r(L_{\infty}).
\end{equation}
\emph{Examples with $r=1$ and $r=2$.}
In case of the uncoloured HOMFLY-PT invariant, the eigenvalues are $\lambda_0=1,\lambda_1=a^2$, and twist recursion is given by
\begin{equation}
    P(K_{i+1})-a^2P(K_i)=P^{\text{fin}}(L_{\infty}).
\end{equation}
By taking $i\to\infty$ limit, we get
\begin{equation}
    (1-a^2)P(K_{\infty})= P^{\text{fin}}(L_{\infty}),
\end{equation}
i.e.
\begin{equation}\label{eq:inf twist to a link}
    P(K_{\infty}) = \frac{P^{\text{fin}}(L_{\infty})}{(a^2;q^2)_1} \simeq \frac{P(L_{\infty})}{\overline{P}(0_1)}.
\end{equation}
For $r=2$, the eigenvalues are $\lambda_0=1,\lambda_1=q^2,\lambda_2=a^4q^4$, and the twist recursion is given by
\begin{equation}
    P_2(K_{j+2})-(a^2+a^4q^4)P_2(K_{j+1})+a^6q^4P_2(K_{j}) = (1+a^2q^2)P_2^{\text{fin}}(L_{\infty}).
\end{equation}
By taking $j\to\infty$ limit, we get
\begin{equation}
    (1-a^2)(1-a^4q^4)P_2(K_{\infty}) = \frac{1-a^4q^4}{1-a^2q^2}P_2^{\text{fin}}(L_{\infty}),
\end{equation}
i.e.
\begin{equation}
    P_2(K_{\infty}) = \frac{P_2^{\text{fin}}(L_{\infty})}{(a^2;q^2)_2} \simeq \frac{P_2(L_{\infty})}{\overline{P}_2(0_1)}.
\end{equation}

    We note that the case of parallel orientation of strands is more subtle, and there is no clear statement about the limit of (\ref{genrp}). Indeed, the nice interpretation of the limit of (\ref{genra}) boils down to interpreting the constant in (\ref{eq:homflypt twist recursion anti-parallel with const}) as a~finite link polynomial. Unfortunately, such a~simple interpretation is not available for the parallel case, at least for an arbitrary sequence of knots $K_i$. Rather, a~(suitable) limit matches the value of the coloured HOMFLY-PT polynomial of a~diagram that is obtained from the knot $K_0$  where we eliminate all the full twists and put the (Jones-Wenzl) symmetrizer over the two parallel strands involved in  the twist region. A simple explicit result exists in the particular case of $(2,2p+1)$ torus knots. In that case the limit of the $Sym^r$ coloured HOMFLY-PT polynomials of $(2,2p+1)$ torus knots, when $p\to \infty$ matches the (up to an overall multiple) $2r$-coloured unknot. We provide a~detailed analysis of this case in Appendix \ref{sec:Stable limit for torus knots}.

\section{Stable growth of quivers}\label{sec:Quiver perspective}

In this Section we analyse the stable growth of symmetric quivers.
We start from recalling definitions of motivic generating series, splitting, unlinking, linking, and then we show how they lead to the stable growth and discuss the limiting behaviour.

\subsection{Prerequisites}\label{sec:Twisting and splitting}

Quiver $Q$ is a~directed graph consisting of a~finite number of nodes and arrows connecting them. We denote the~number of nodes by $m$ and assemble the numbers of arrows into $m\times m$ adjacency matrix $C$: the number of arrows from node $i$ to node $j$ is given by $C_{ij}$. A~quiver is called symmetric if for each arrow between two different vertices there is also an~arrow in the~opposite direction; this means that $C_{ij}=C_{ji}$. 

\subsubsection{Motivic generating series and DT invariants}

For a~given quiver, it is important to understand the~structure of moduli spaces of its representations. (Quiver representation is an assignment of a~vector space $\mathbb{C}^{d_i}$ to vertex~$i$ and a~linear map $\mathbb{C}^{d_i}\to\mathbb{C}^{d_j}$ to quiver arrow $i\to j$ .) Basic information about topology of such spaces is encoded in their Betti numbers (or their generalisations), which are captured by the motivic Donaldson-Thomas (DT) invariants $\Omega_{\boldsymbol{d},s}$ which are indexed by the~dimension vector $\boldsymbol{d}=(d_1,\ldots,d_{m})\in \mathbb{N}^{m}$ and integer $s$. For a~symmetric quiver, these invariants are encoded in the~motivic generating series
\begin{equation}\label{eq:quiver_gen_series}
    P_Q(\boldsymbol{x},q) =
    \sum_{\boldsymbol{d}}(-q)^{\boldsymbol{d} \cdot C\cdot\boldsymbol{d}}\frac{\boldsymbol{x}^{\boldsymbol{d}}}{(q^{2};q^{2})_{\boldsymbol{d}}}
    = \sum_{d_1,\dots,d_{m}\geq 0}(-q)^{\sum_{i,j=1}^{m} C_{ij}d_i d_j}\prod_{i=1}^{m}\frac{x_i^{d_i}}{(q^{2};q^{2})_{d_i}},
\end{equation}
where a~generating parameter $x_i$ is assigned to each vertex $i$ and 
\begin{equation}
    (\xi;q^{2})_{n} = \prod_{k=0}^{n-1}(1-\xi q^{2k})
\end{equation}
is the~$q$-Pochhammer symbol. The~product decomposition of this series into quantum dilogarithms determines DT invariants as follows:
\begin{equation}\label{PQ-product}
    P_Q(\boldsymbol{x},q) = \prod_{\boldsymbol{d},s}(\boldsymbol{x}^{\boldsymbol{d}}q^s;q^2)_{\infty}^{\Omega_{\boldsymbol{d},s}} = \prod_{\boldsymbol{d}\in \mathbb{N}^{m}\setminus \boldsymbol{0}} \prod_{s\in\mathbb{Z}} \prod_{k\geq 0} \Big(1 - (x_1^{d_1}\cdots x_{m}^{d_{m}}) q^{2k+s} \Big)^{\Omega_{(d_1,\ldots,d_{m}),s}}.
\end{equation}
It is convenient to combine $\Omega_{\boldsymbol{d},s}$ into a~generating series:
\begin{equation}
    \Omega_Q(\boldsymbol{x},q) = \sum_{\boldsymbol{d}} \Omega_{\boldsymbol{d}}(q)\, \boldsymbol{x}^{\boldsymbol{d}} =  \sum_{\boldsymbol{d}\in \mathbb{N}^{m}\setminus \boldsymbol{0}}\sum_{s\in \mathbb{Z}} \,\Omega_{(d_1,\ldots,d_{m}),s}\, x_1^{d_1}\cdots x_{m}^{d_{m}} q^s.   \label{Omega-series}
\end{equation}
DT invariants are integer numbers \cite{KS1006,Efi12}:
\begin{equation}
    \Omega_{\boldsymbol{d},s}\in \mathbb{Z}\,,
\end{equation}
which is a~reflection of their physical interpretation as number of BPS states in 3d $\mathcal{N}=2$ theory analysed in~\cite{EKL1811}.\footnote{Note that in this paper the~sign is included in the~definition of $\Omega$'s.
In consequence, we have $\sum_{s\in\mathbb{Z}}\Omega_{\boldsymbol{d},s}q^s=\sum_{j\in\mathbb{Z}}\Omega'_{\boldsymbol{d},j}(-q)^{j+1}$, where $\Omega'_{\boldsymbol{d},j}\geq 0$ is the~DT invariant in the~notation from \cite{KRSS1707long}.} We denote the set of BPS states for quiver $Q$ by $\Theta_Q$ and analyse its structure in Section \ref{sec:Limiting behaviour for infinite splitting or unlinking}.

Our focus on the motivic generating series $P_Q(\boldsymbol{x},q)$ means that it is convenient to use the notation $Q=(C,\boldsymbol{x})$ to represent a~quiver via its adjacency matrix, together with the corresponding generating parameters. Moreover, the interpretation of $C_{ij}$ as the number of arrows implies that $C_{ij}\in \mathbb{N}$, but we consider more general case of $C_{ij}\in \mathbb{Z}$ -- in that case we simply treat them as parameters in (\ref{eq:quiver_gen_series}) or, equivalently, use \cite[Definition 3]{JKLNS2212}.

\subsubsection{Splitting}\label{sec:splitting}

Splitting is a~transformation of the quiver and its generating parameters which was introduced in~\cite{JKLNS2105} to explain how different quivers can lead to the same motivic generating series and why they form structures based on permutations. We will use it to describe stable growth of quivers.

Consider quiver $Q=(C,\boldsymbol{x})$ with $m$ nodes and pick $n$ of them together with a~permutation $\sigma\in S_n$. We call the remaining $m-n$ nodes spectators and associate an integer shift $h_s$ to each spectator $s$. We also pick  a~multiplicative factor $\kappa$ that will govern the behaviour of generating parameters.

\begin{defn}\label{def:splitting}  
    A~$(k,l)$-splitting of a~quiver $Q$  is defined as the~following transformation of adjacency matrix $C$ and vector of generating parameters $\boldsymbol{x}$. For any two nodes $i$ and $j$ that we split (without loss of generality we assume $i<j$) and any spectator node $s$ we transform
    \begin{equation}
       C=\left[\begin{array}{c:c:c:c:c}
C_{ss} & \cdots  &C_{si} & \cdots  &C_{sj}\\ \hdashline
\vdots  & \ddots  & \vdots  &  & \vdots \\ \hdashline
C_{is} & \cdots  &C_{ii} & \cdots  &C_{ij}\\ \hdashline
\vdots  &  & \vdots  & \ddots  & \vdots \\ \hdashline
C_{js} & \cdots  &C_{ji} & \cdots  &C_{jj}
\end{array}\right]\,
    \end{equation}
    depending on the~presence of inversion in permutation $\sigma$. For $\sigma(i)<\sigma(j)$ we transform $C$ into
    \begin{equation}
            \left[\begin{array}{ c : c : c c : c : c  c }
C_{ss} & \cdots  & C_{si} & C_{si} +h_{s} & \cdots  & C_{sj} & C_{sj} +h_{s}\\ \hdashline
\vdots  & \ddots  & \vdots  & \vdots  &  & \vdots  & \vdots \\ \hdashline
C_{is} & \cdots  & C_{ii} & C_{ii} +k & \cdots  & C_{ij} & \textcolor[rgb]{0.96,0.65,0.14}{C_{ij} +k}\\ 
C_{is} +h_{s} & \cdots  & C_{ii} +k & C_{ii} +l & \cdots  & \textcolor[rgb]{0.96,0.65,0.14}{C_{ij} +k+1} & C_{ij} +l\\ \hdashline
\vdots  &  & \vdots  & \vdots  & \ddots  & \vdots  & \vdots \\ \hdashline
C_{js} & \cdots  & C_{ji} & \textcolor[rgb]{0.96,0.65,0.14}{C_{ji} +k+1} & \cdots  & C_{jj} & C_{jj} +k\\ 
C_{js} +h_{s} & \cdots  & \textcolor[rgb]{0.96,0.65,0.14}{C_{ji} +k} & C_{ji} +l & \cdots  & C_{jj} +k & C_{jj} +l
\end{array}\right]\,,
    \end{equation}
whereas  for $\sigma(i)>\sigma(j)$ we transform $C$ into
    \begin{equation}
      \left[\begin{array}{c:c:cc:c:cc}
C_{ss} & \cdots  & C_{si} & C_{si} +h_{s} & \cdots  & C_{sj} & C_{sj} +h_{s}\\ \hdashline
\vdots  & \ddots  & \vdots  & \vdots  &  & \vdots  & \vdots \\ \hdashline
C_{is} & \cdots  & C_{ii} & C_{ii} +k & \cdots  & C_{ij} & \textcolor[rgb]{0.96,0.65,0.14}{C_{ij} +k+1}\\ 
C_{is} +h_{s} & \cdots  & C_{ii} +k & C_{ii} +l & \cdots  & \textcolor[rgb]{0.96,0.65,0.14}{C_{ij} +k} & C_{ij} +l\\ \hdashline
\vdots  &  & \vdots  & \vdots  & \ddots  & \vdots  & \vdots \\ \hdashline
C_{js} & \cdots  & C_{ji} & \textcolor[rgb]{0.96,0.65,0.14}{C_{ji} +k} & \cdots  & C_{jj} & C_{jj} +k\\ 
C_{js} +h_{s} & \cdots  & \textcolor[rgb]{0.96,0.65,0.14}{\check{\textcolor[rgb]{0.96,0.65,0.14}{C}}_{ji} +k+1} & C_{ji} +l & \cdots  & C_{jj} +k & C_{jj} +l
\end{array}\right]\,.  
    \end{equation}
For any permutation the generating parameters are transformed as follows:
\begin{equation*}
    \left[\begin{array}{c}
    x_s\\ \hdashline
    \vdots\\ \hdashline
    x_i\\ \hdashline
    \vdots\\ \hdashline
    x_j
    \end{array}\right]
    \longrightarrow
    \left[\begin{array}{c}
    x_s\\ \hdashline
    \vdots\\ \hdashline
    x_i\\
    x_i\kappa\\ \hdashline
    \vdots\\ \hdashline
    x_j\\
    x_j\kappa
    \end{array}\right]\,.
\end{equation*}
\end{defn}

\subsubsection{Unlinking and linking}

Another transformations of quivers that are important in the stable growth are unlinking and linking. Linking and unlinking respectively add or remove one pair of arrows between two particular nodes $i$ and $j$, as well as increase the size of a~quiver by one node, which gets connected by a~particular pattern of arrows to other nodes. They were introduced in \cite{EKL1910}, together with an~interesting interpretation in terms of multi-cover skein relations and 3d $\mathcal{N}=2$ theories discussed earlier in \cite{EKL1811,ES1901}.  

\begin{defn}\label{defn:symmetric quiver notation}
    For symmetric quiver $Q$, given by the following adjacency matrix and vector of generating parameters\footnote{For brevity and clarity of the presentation we write only the upper-triangular part of the symmetric matrix and use $\cdot$ instead of $\cdots$ and $\ddots$ to denote multiple entries in the matrix.}
    \begin{equation}
    Q = \;\; \left[\begin{array}{ccccccc}
 C_{11}  &  \cdot  &  C_{1i}  &   \cdot  &  C_{1j}  &   \cdot  &  C_{1m}\\
   &  \cdot\  &  \vdots &   &  \vdots  &   &  \vdots\\
 &  &  C_{ii}  &   \cdot &  C_{ij}  &   \cdot &  C_{im}\\
 &  &  &  \cdot\  &  \vdots  &   &  \vdots\\
 &  &  &  &  C_{jj}  &   \cdot  &  C_{jm}\\
 &  &  &  &  &  \cdot &  \vdots\\
 &  &  &  &  &  & C_{mm}
\end{array}\right],\;\; \left[\begin{array}{c} x_{1} \\ \vdots \\ x_{i} \\ \vdots \\ x_{j} \\ \vdots \\ x_{m}
\end{array}\right]\;\;
\end{equation}
we define unlinking of distinct nodes $i$ and $j$ as operator $U(ij)$ acting on $Q$ as follows:
    \begin{equation}\label{eq:unlinking definition}
    U(ij)Q = \;\; \left[\begin{array}{cccccccc}
C_{11}  &   \cdot &  C_{1i}  &   \cdot &  C_{1j}  &   \cdot &  C_{1m}  &  C_{1i}+C_{1j}\\
 &  \cdot\  &  \vdots &   &  \vdots &   &  \vdots &  \vdots\\
 &  &  C_{ii}  &   \cdot &  C_{ij}-1  &   \cdot &  C_{im}  &  C_{ii}+C_{ij}-1\\
 &  &  &  \cdot\  &  \vdots &   &  \vdots &  \vdots\\
 &  &  &  &  C_{jj}  &   \cdot &  C_{jm}  &  C_{ij}+C_{jj}-1\\
 &  &  &  &  &  \cdot\  &  \vdots &  \vdots\\
 &  &  &  &  &  &  C_{mm}  &  C_{im}+C_{jm}\\
 &  &  &  &  &  &    &  C_{ii}+C_{jj}+2C_{ij}-1
\end{array}\right],\;\; \left[\begin{array}{c} x_{1} \\ \vdots \\ x_{i} \\ \vdots \\ x_{j} \\ \vdots \\ x_{m} \\ q^{-1}x_{i}x_{j}
\end{array}\right]\;\;
\end{equation}
We define linking of distinct nodes $i$ and $j$ as operator $L(ij)$ acting on $Q$ as follows:
    \begin{equation}\label{eq:linking definition}
    L(ij)Q = \;\; \left[\begin{array}{cccccccc}
C_{11}  &  \cdot &  C_{1i}  &  \cdot &  C_{1j}  &  \cdot &  C_{1m}  &  C_{1i}+C_{1j}\\
 &  \cdot\  &  \vdots &   &  \vdots &   &  \vdots &  \vdots\\
 &  &  C_{ii}  &  \cdot &  C_{ij}+1  &  \cdot &  C_{im}  &  C_{ii}+C_{ij}\\
 &  &  &  \cdot\  &  \vdots &   &  \vdots &  \vdots\\
 &  &  &  &  C_{jj}  &  \cdot &  C_{jm}  &  C_{ij}+C_{jj}\\
 &  &  &  &  &  \cdot\  &  \vdots &  \vdots\\
 &  &  &  &  &  &  C_{mm}  &  C_{im}+C_{jm}\\
 &  &  &  &  &  &    &  C_{ii}+C_{jj}+2C_{ij}
\end{array}\right],\;\; \left[\begin{array}{c} x_{1} \\ \vdots \\ x_{i} \\ \vdots \\ x_{j} \\ \vdots \\ x_{m} \\ x_{i}x_{j}
\end{array}\right]\;\;
\end{equation}
\end{defn}

In \cite{EKL1910} it was proven that unlinking and linking preserves the motivic generating series:
\begin{equation}
     P_Q(\boldsymbol{x},q) =  P_{U(ij)Q}(\boldsymbol{x},q)  =  P_{L(ij)Q}(\boldsymbol{x},q), 
\end{equation}
where $P_{U(ij)Q}(\boldsymbol{x},q)$ and $P_{L(ij)Q}(\boldsymbol{x},q)$ denote the motivic generating series corresponding to the matrices and vectors of generating parameters given by \eqref{eq:unlinking definition}
 and \eqref{eq:linking definition} respectively\footnote{Note that in all cases the series depends on $m$ generating parameters $x_1,x_2,\dots, x_m$.}.

\subsection{Splitting by unlinking and linking}

Transformations of quivers presented previously are related. It was proven in  \cite{KLNS2312}, basing on the idea from  \cite{EKL2108}, that unlinking can be used to model splitting:

\begin{proposition}
\cite[Proposition 6.1]{KLNS2312}
\label{prp:splitting by unlinking}
Every splitting of quiver $Q$ can be realised by an application of a~sequence of unlinkings if we allow for the presence of an extra node. 
\end{proposition}
\begin{defn}\label{defn:augmented quiver}
We denote the quiver with that extra node $Q^{+}$ and call it
the \emph{augmented quiver}. 
\end{defn}

Below we provide an explicit example of $(0,2)$-splitting for an arbitrary symmetric quiver $Q$ with three nodes. For this, we perform the splitting on the second and third node: 

\begin{equation}\renewcommand\arraystretch{1.2
}
\left[
\begin{array}{ccc}
 C_{11} & C_{12} & C_{13} \\
 C_{12} & C_{22} & C_{23} \\
 C_{13} & C_{23} & C_{33} \\
\end{array}
\right]
\overset{(0,2)}{\longrightarrow}
    \left[
\begin{array}{ccccc}
 C_{11} & C_{12} & C_{12}+h_1 & C_{13} & C_{13}+h_1 \\
 C_{12} & C_{22} & C_{22}+0 & C_{23} & C_{23}+0 \\
 C_{12}+h_1 & C_{22}+0 & C_{22}+2 & C_{23}+(0+1) & C_{23}+2 \\
 C_{13} & C_{23} & C_{23}+(0+1) & C_{33} & C_{33}+0 \\
 C_{13}+h_1 & C_{23}+0 & C_{23}+2 & C_{33}+0 & C_{33}+2 \\
\end{array}
\right],\label{matrix-split}
\end{equation}
where $h_1$ is some integer. Now, we can also obtain this $(0,2)$-splitting of $Q$ by using unlinking of the augmented quiver $Q^+$ which has an extra node,\footnote{We use index $0$ to denote this extra node.} as shown below (for clarity, we always separate the augmentation node with lines). Applying $U(03)U(02)Q^+$, we get

\begin{equation}\renewcommand\arraystretch{1.2
}
    U(03)U(02)\left[
\begin{array} {c|ccc}
1 & 0 & 1 & 1 \\
\hline
0 & C_{11} & C_{12} & C_{13} \\
1 & C_{12} & C_{22} & C_{23} \\
1 & C_{13} & C_{23} & C_{33} \\
\end{array}
\right]=\left[
\begin{array}{c|ccccc}
1 & 0 & 0 & 0 & 1 & 1 \\
\hline
0 & C_{11} & C_{12} & C_{13} & C_{12} & C_{13} \\
0 & C_{12} & C_{22} & C_{23} & C_{22} & C_{23} \\
0 & C_{13} & C_{23} & C_{33} & C_{23} + 1 & C_{33} \\
1 & C_{12} & C_{22} & C_{23} + 1 & C_{22} + 2 & C_{23} + 2 \\
1 & C_{13} & C_{23} & C_{33} & C_{23} + 2 & C_{33} + 2 \\
\end{array}
\right].\label{matrix-unlink}
\end{equation}

Note that the integer shift $h_1=0$ in this case. Moreover, the quivers given by (\ref{matrix-split}) and (\ref{matrix-unlink}) are the same after removing the extra node from (\ref{matrix-unlink}). This can be seen easily if we shuffle the third and fourth diagonal entries in (\ref{matrix-unlink}). The vector of variables under unlinking transforms as follows:
\begin{equation}
\left[
\begin{array}{c}
 x_0 \\
 \hline
 x_1 \\
 x_2 \\
 x_3 \\
\end{array}
\right]\longrightarrow\left[
\begin{array}{c}
 x_0 \\
 \hline
 x_1 \\
 x_2 \\
 x_3 \\
q^{-1}x_0 x_2 \\
 q^{-1}x_0x_3 \\
\end{array}
\right],
\end{equation}
where $x_0$ corresponds to the extra node.


It turns out that linking also can be used to model splitting (The proof of this proposition can be found in Appendix \ref{app:splitting_by_linking}.):
\begin{proposition}\label{prp:splitting_by_linking}
Every splitting can be realised by an application of a~sequence of linkings to augmented quiver.
\end{proposition}

For example, $(-1,-2)$-splitting of arbitrary symmetric quiver with three nodes is given by
\begin{equation}\renewcommand\arraystretch{1.2
}
\left[
\begin{array}{ccc}
 C_{11} & C_{12} & C_{13} \\
 C_{12} & C_{22} & C_{23} \\
 C_{13} & C_{23} & C_{33} \\
\end{array}
\right]
\overset{(-1,-2)}{\longrightarrow}
    \left[
\begin{array}{ccccc}
 C_{11} & C_{12} & C_{12}+h_1 & C_{13} & C_{13}+h_1 \\
 C_{12} & C_{22} & C_{22}-1 & C_{23} & C_{23}-1 \\
 C_{12}+h_1 & C_{22}-1 & C_{22}-2 & C_{23}+(-1+1) & C_{23}-2 \\
 C_{13} & C_{23} & C_{23}+(-1+1) & C_{33} & C_{33}-1 \\
 C_{13}+h_1 & C_{23}-1 & C_{23}-2 & C_{33}-1 & C_{33}-2 \\
\end{array}
\right].\label{matrix-split-2}
\end{equation}
Its realisation by linking with an extra node can be written as (note the difference in the linking numbers for the extra node between $(0,2)$- and $(-1,-2)$-splitting):
\begin{equation}\renewcommand\arraystretch{1.2
}
   L(02)L(03)\left[
\begin{array}{c|ccc}
0 & 0 & -1 & -1 \\
\hline
0 & C_{11} & C_{12} & C_{13} \\
-1 & C_{12} & C_{22} & C_{23} \\
-1 & C_{13} & C_{23} & C_{33} \\
\end{array}
\right]=\left[
\begin{array}{c|ccccc}
0 & 0 & 0 & 0 & - 1 & - 1 \\
\hline
0 & C_{11} & C_{12} & C_{13} & C_{13} & C_{12} \\
0 & C_{12} & C_{22} & C_{23} & C_{23} - 1 & C_{22} - 1 \\
0 & C_{13} & C_{23} & C_{33} & C_{33} - 1 & C_{23} \\
-1 & C_{13} & C_{23} - 1 & C_{33} - 1 & C_{33} - 2 & C_{23} - 2 \\
-1 & C_{12} & C_{22} - 1 & C_{23} & C_{23} - 2 & C_{22} - 2 \\
\end {array}
\right].\label{matrix-link}
\end{equation}
Again, the quivers given by (\ref{matrix-split-2}) and (\ref{matrix-link}) are same after removing the extra node from (\ref{matrix-link}). In order to see this, we need to reorder the last three diagonal entries in (\ref{matrix-link}). In this case, the vector of variables transforms as follows:
\begin{equation}
\left[
\begin{array}{c}
 x_0 \\
 \hline
 x_1 \\
 x_2 \\
 x_3 \\
\end{array}
\right]\longrightarrow\left[
\begin{array}{c}
 x_0 \\
 \hline
 x_1 \\
 x_2 \\
x_3 \\
x_0 x_3 \\
x_0x_2 \\
\end{array}
\right].
\end{equation}
It is worth noting here that the order of linking or unlinking operation results in quivers which have the same generating series, but differ by permutation of entries. Detailed discussion of this phenomenon is available in \cite{JKLNS2105,KLNS2312}.

\subsection{Stability and limiting behaviour of splitting, unlinking, and linking}\label{sec:Limiting behaviour for infinite splitting or unlinking}

In Section \ref{sec:Stable growth for knots} we analysed the stable growth of knot polynomials with twisting. In the remaining part of this Section we study when similar stability can be observed for quivers.

\subsubsection{Stability}

The main idea for describing stable growth of quivers is a~recursive application of transformations of quivers:

\begin{defn}
    Consider a~quiver $Q_{1}$ and pick $n$ nodes. Perform $(k,l)$-splitting of these nodes and call the resulting quiver $Q_{2}$. Continue recursively using the created nodes in the new iteration of $(k,l)$-splitting. We say that quivers $\{Q_{i}\}_{i\in\mathbb{Z}_{+}}$ constructed in a~way described above \emph{grow by $(k,l)$-splitting.}
\end{defn}

Since we know that splitting can be modelled by unlinking we may ask when quivers $\{Q_{i}\}_{i\in\mathbb{Z}_{+}}$ that grow by  $(k,l)$-splitting can be obtained by recursive linking and unlinking in the presence of the extra node.

\begin{proposition}\label{prp:recursive_unlinking}
    Quivers that grow by $(k,2k+2)$-splitting can be obtained by recursive unlinking in the presence of the extra node.
\end{proposition}

\begin{proof}
    From \cite[Proposition 6.1]{KLNS2312} we know that $(k,l)$-splitting of $n$ nodes of quiver $Q$ with permutation $\sigma\in S_{n}$ in the presence of $m-n$ spectators with corresponding shifts $h_{1},\dots,h_{m-n}$, and multiplicative factor $\kappa$ can be modelled by unlinking of a~quiver $Q^{+}$ with the following adjacency matrix\footnote{We rearranged rows and columns to allow for convenient iterations.}: 
    \begin{equation}
    C^{+}=\left[\begin{array}{ccccccc}
     l-2k-1 & h_{1} &\dots  &h_{m-n} &k+1  & \dots &  k+1\\
     &  &  &  &  &  & \\
     &  &  &  &  &  & \\
     &  &  & C &  &  & \\
     &  &  &  &  &  & \\
     &  &  &  &  &  &\\
    &  &  &  &  &  & 
    \end{array}\right]  
    \end{equation}
    After unlinkings, the extra row of the adjacency matrix of $U(m-n+\sigma(n),\iota)\dots U(m-n+\sigma(1),\iota)Q^{+}$ is given by
    \begin{equation}
    \label{eq:extra row/column}
    (\,l-2k-1\,,h_{1}\,,\,\dots,h_{m-n}\,,\,k\,,\,\dots\,,\,k\,,\,l-k-1\,,\,\dots\,,\,l-k-1\,)\,.
    \end{equation}
    Since we want to perform $(k,l)$-splitting again -- but this time applying it to the new $n$ nodes -- we require that \eqref{eq:extra row/column} is equal to\footnote{This time we treat $h_{1}\,,\,\dots,h_{m-n}\,,\,k\,,\,\dots\,,\,k$ as shifts corresponding to spectators.}
 \begin{equation}
    (\,l-2k-1\,,h_{1}\,,\,\dots,h_{m-n}\,,\,k\,,\,\dots\,,\,k\,,\,k+1\,,\,\dots\,,\,k+1\,)\,,
    \end{equation}
    which implies that  $l-k-1=k+1$. In consequence, above procedure may be used to construct quivers that grow by $(k,2k+2)$-splitting.
\end{proof}

\begin{proposition}\label{prp:recursive_linking}
    Quivers that grow by $(k,2k)$-splitting can can be obtained by recursive linking in the presence of the extra node.
\end{proposition}

\begin{proof}
    From the proof of Proposition \ref{prp:splitting_by_linking} we know that $(k,l)$-splitting of $n$ nodes of quiver $Q$ with permutation $\sigma\in S_{n}$ in the presence of $m-n$ spectators with corresponding shifts $h_{1},\dots,h_{m-n}$, and multiplicative factor $\kappa$ can be modelled by linking of a~quiver $Q^{+}$ with the following adjacency matrix:
    \begin{equation}
    C^{+}=\left[\begin{array}{ccccccc}
     l-2k & h_{1} &\dots  &h_{m-n} &k  & \dots &  k\\
     &  &  &  &  &  & \\
     &  &  &  &  &  & \\
     &  &  & C &  &  & \\
     &  &  &  &  &  & \\
     &  &  &  &  &  &\\
    &  &  &  &  &  & 
    \end{array}\right]  
    \end{equation}
    After linkings, the extra row of the adjacency matrix of $L(m-n+\sigma(1),\iota)\dots L(m-n+\sigma(n),\iota)Q^{+}$ is given by
    \begin{equation}
    \label{eq:extra row/column_linking}
    (\,l-2k\,,h_{1}\,,\,\dots,h_{m-n}\,,\,k+1\,,\,\dots\,,\,k+1\,,\,l-k\,,\,\dots\,,\,l-k\,)\,.
    \end{equation}
    Since we want to perform $(k,l)$-splitting again -- but this time applying it to the new $n$ nodes -- we require that \eqref{eq:extra row/column_linking} is equal to
 \begin{equation}
    (\,l-2k\,,h_{1}\,,\,\dots,h_{m-n}\,,\,k+1\,,\,\dots\,,\,k+1\,,\,k\,,\,\dots\,,\,k\,)\,,
    \end{equation}
    which implies that  $l-k=k$. In consequence, above procedure may be used to construct quivers that grow by $(k,2k)$-splitting.
\end{proof}

\subsubsection{Limiting behaviour}

From the quiver perspective, the limiting behaviour of the stable growth can be presented in a~very elegant way in terms of infinite quivers (one can also think of it as a~guiding principle for the definition of the stable growth of quivers).
\begin{proposition}\label{prp:stable_growth_of_quivers}
For quivers $\{Q_{i}\}_{i\in\mathbb{Z}_{+}}$ that grow by $(k,2k+2)$-splitting or $(k,2k)$-splitting of $n$ nodes there exist an infinite quiver $Q_{\infty}$ such that the adjacency matrix of $Q_{i}$ is given by the top-left subquiver of the size $\left(m_1+ni\right)\times\left(m_1+ni\right)$, where $m_1$ is the number of nodes in $Q_1$.
\end{proposition}

\begin{proof}
Since splitting enlarges the quiver keeping the initial nodes and arrows intact, we know that $Q_{i}\subset Q_{i+1}$. If we arrange each $C_i$ -- the adjacency matrix of  $Q_{i}$ -- in such a~way that new nodes coming from splitting are the rightmost (like in Propositions \ref{prp:recursive_unlinking} and \ref{prp:recursive_linking}) , then $C_{i+1}$  consists of $C_i$ in the top-left corner and their numbers of nodes are related by $m_{i+1}=m_i+n$. This means that $Q_\infty = \bigcup_{i=1}^{\infty}Q_i$  satisfies the requirements of the proposition.
\end{proof}

In a~special cases of quivers that grow by $(0,2)$-splitting or $(-1,-2)$-splitting, we can describe the spectrum of BPS states of each of these infinite families using a~single augmented quiver.
\begin{thm}\label{thm:BPS_spectrum_of_augmented_quiver}
For quivers $\{Q_{i}\}_{i\in\mathbb{Z}_{+}}$ that grow by $(0,2)$-splitting or $(-1,-2)$-splitting of $n$ nodes in the presence of spectators with corresponding shifts equal to $0$, the spectrum of BPS states of $Q_{\infty} = \bigcup_{i=1}^{\infty}Q_i$ is equal to the spectrum of the augmented quiver $Q_{1}^{+}$ minus the single state coming from the extra node.
\end{thm}

\begin{proof}
Quivers $\{Q_{i}\}_{i\in\mathbb{Z}_{+}}$ that grow by $(0,2)$-splitting in the presence of spectators with corresponding shifts equal to $0$ can be modelled by recursive unlinking of an~augmented quiver $Q_1^{+}$ with the following adjacency matrix: 
    \begin{equation}
    Q_1^{+}=\left[\begin{array}{ccccccc}
     1 & 0 &\dots  &0 &1  & \dots &  1\\
     &  &  &  &  &  & \\
     &  &  &  &  &  & \\
     &  &  & Q_1 &  &  & \\
     &  &  &  &  &  & \\
     &  &  &  &  &  &\\
    &  &  &  &  &  & 
    \end{array}\right]  \,.
    \end{equation}
In the infinite limit we have    
\begin{equation}
    Q_\infty^{+}=\left[\begin{array}{ccccccc}
     1 & 0 & 0  & 0 & 0  & 0 &  \dots\\
     &  &  &  &  &  & \\
     &  &  &  &  &  & \\
     &  &  & Q_\infty &  &  & \\
     &  &  &  &  &  & \\
     &  &  &  &  &  &\\
    &  &  &  &  &  & 
    \end{array}\right]  \,,
    \end{equation}
which implies that the spectrum of BPS states of $Q_\infty^{+}$ is equal to the one of $Q_\infty$ plus a~single state coming from the detached extra node with one loop. Moreover, since unlinking preserves the motivic generating series, we have
\begin{equation}
    P_{Q_1^{+}}(\boldsymbol{x},q) =  P_{Q_\infty^{+}}(\boldsymbol{x},q) \,,
\end{equation}
  therefore the spectrum of BPS states of  $Q_\infty$ is equal to the one of $Q_1^{+}$ minus a~single state coming from the extra node.

The proof for quivers $\{Q_{i}\}_{i\in\mathbb{Z}_{+}}$ that grow by $(-1,-2)$-splitting is completely analogous, but starts from an~augmented quiver $Q_1^{+}$ with the following adjacency matrix:
    \begin{equation}
    Q_1^{+}=\left[\begin{array}{ccccccc}
     0 & 0 &\dots  &0 &-1  & \dots &  -1\\
     &  &  &  &  &  & \\
     &  &  &  &  &  & \\
     &  &  & Q_1 &  &  & \\
     &  &  &  &  &  & \\
     &  &  &  &  &  &\\
    &  &  &  &  &  & 
    \end{array}\right] \,,
    \end{equation}
 uses linking instead of unlinking, and the extra node has no loops instead of one.
\end{proof}

Let us analyse $\Theta_{Q_i}$ -- the set of BPS states for quiver $Q_i$ which are counted by DT invariants $\Omega_{Q_i}(\boldsymbol{x},q)$ -- for all $i\in \mathbb{Z_+}$. Since  $Q_{i}\subset Q_{i+1}$, we know that $\Theta_{Q_{i}}\subset \Theta_{Q_{i+1}}$. Considering the infinite limit, we can see that
\begin{equation}\label{eq:DT_subsets}
    \Theta_{Q_{1}}\subset \Theta_{Q_{2}}\subset \dots \subset\Theta_{Q_{\infty}} \,.
\end{equation}
On the other hand, invariance under unlinking and linking guarantees that
\begin{equation}\label{eq:augmented_DT_subsets}
    \Theta_{Q^+_{1}} = \Theta_{Q^+_{2}} = \dots = \Theta_{Q^+_{\infty}} \,.
\end{equation}
Theorem \ref{thm:BPS_spectrum_of_augmented_quiver} allows us to combine \eqref{eq:DT_subsets}   and \eqref{eq:augmented_DT_subsets} into
\begin{equation}
\begin{matrix}    
        \Theta_{Q_{1}} & \subset & \Theta_{Q_{2}} & \subset & \dots & \subset & \Theta_{Q_{\infty}} &  & \\
        \cap &  & \cap &  &  &  & \cap &  & \\
        \Theta_{Q^+_{1}} & = & \Theta_{Q^+_{2}} & = & \dots & = & \Theta_{Q^+_{\infty}} & = & \Theta_{Q_{\infty}}\cup \{\text{state from the extra node}\}\\
\end{matrix}    
\end{equation}
In consequence, the BPS spectrum of the infinite family of quivers that grow by $(0,2)$- or $(-1,-2)$-splitting is captured by single finite (and relatively small) quiver $Q^+_{1}$. Moreover, extracting $\Theta_{Q_{i}}$ from $\Theta_{Q^+_{i}}=\Theta_{Q^+_{1}}$ boils down to setting $x_0=0$ in $\Omega_{Q^+_{i}}(\boldsymbol{x},q)$, and $Q^+_{i}$ can be obtained from $Q^+_{1}$ by $i-1$  unlinkings (or linkings). The bigger $i$, the less BPS states we have to subtract from  $\Theta_{Q^+_{i}}=\Theta_{Q^+_{1}}$  and finally at the infinite limit this is only a~single state coming from the detached extra node.
Examples of this behaviour are studied in detail in Sections \ref{sec:Twist_knots}  (the case of growth by $(0,2)$--splitting) and \ref{sec:knot_complements} (growth by $(-1,-2)$-splitting).

\section{Knot-quiver stable growth conjecture}\label{sec:knot-quiver_growth}

In this Section we introduce our main conjecture which makes use of the knot-quiver correspondence in order to unite Sections \ref{sec:Stable growth for knots} and \ref{sec:Quiver perspective} into a~single perspective. 

\subsection{Prerequisites}

\subsubsection{Knot-quiver correspondence}

Let $K\subset S^3$ be a~knot. Consider the HOMFLY-PT generating series:
\begin{equation}\label{eq:homfly gen series}
    P_K(x,a,q)=\sum_{r=0}^{\infty}\frac{P_{K,r}(a,q)}{(q^2;q^2)_r}x^r,
\end{equation}
where $P_{K,r}(a,q)$ are symmetrically coloured HOMFLY-PT polynomials of $K$. The knot-quiver correspondence relates generating series (\ref{eq:homfly gen series}) and the motivic generating series of a~symmetric quiver \cite{KRSS1707short,KRSS1707long}:
    \begin{equation}\label{eq:KQ}
        P_K(x,a,q) = \left. P_{Q_K}(\boldsymbol{x},q) \right|_{x_i=x a^{a_i} q^{q_i}}\,,
    \end{equation}
where $a_i,q_i$ are fixed integers.
Note that the analogous statement also exists for coloured Poincaré polynomials.
To date, this conjecture is confirmed for all rational and, more generally, arborescent knots and links \cite{SW-I,SW-II}. A more general version of the correspondence is also conjectured for an arbitrary knot or a link, which also includes ``higher level'' generators with $x_i=x^{\mu_i}a^{a_i}q^{q_i}$, $\mu_i\geq1$ \cite{EKL2108,Stosic:2024egq}. (In our work, however, we focus on the simplest case with $\mu_i=1$.)

Symmetric quivers which correspond to knots are known to have very rich combinatorial structure \cite{JKLNS2105,KLNS2312}. We follow the same notation for quivers $Q_K$ as in Definition \ref{defn:symmetric quiver notation}, keeping in mind that the variables $x_i$ are now monomials in $a,q,x$, as required by (\ref{eq:KQ}). One can also introduce the analogue of DT invariants for knots:
\begin{equation}\label{eq:P_K=00003DExp}
P_{K}(x,a,q)=\prod_{r\geq 1}\prod_{i,j\in \mathbb{Z}}(x^{r} a^i q^j;q^2)_{\infty}^{N_{r,i,j}^K}\,.
\end{equation}
The numbers $N_{r,i,j}^K\in\mathbb{N}$ count degeneracies of BPS states in 3d $\mathcal{N}=2$ theory $T[L_K]$ associated to a~knot $K$. Introduced by Ooguri and Vafa \cite{OV9912}, these numbers were further studied by Labastida, Mariño and Vafa \cite{LM01,LM02,LMV00} and are commonly known as LMOV invariants. Using the knot-quiver correspondence, they can be expressed in terms of DT invariants of the~corresponding quiver, which also confirms their integrality:
\begin{equation}\label{eq:LMOV vs DT}
N_K(x,a,q)=\left.\Omega_{Q_K}(\boldsymbol{x},q)\right|_{x_i=x a^{a_i} q^{q_i}}.
\end{equation}
Analogous invariants can be introduced not only for knots, but also for knot complements and any other adjacency matrices accompanied by knot-quiver change of variables.

\subsubsection{Generalisation to knot complements}\label{knot-complement-quiver-section}

$F_K$ invariant is defined to be an invariant arising from 3d $\mathcal{N}=2$ $T[S^3\backslash K]$ theory associated to a~knot complement $S^3\backslash K$ \cite{GM1904}. This invariant is a~series in two variable $x$ and $q$. Later, in \cite{Ekholm:2021irc} this invariant was refined to give three variable $F_K(x,a,q)$ invariant.
In light of the knot-quiver correspondence, a~relation between knot complements and quiver representation theory was proposed in \cite{Kuch2005}.  The author conjectured that $F_K$ invariants can be expressed by the motivic generating series of a symmetric quiver, in analogy to \eqref{eq:KQ}:
	\begin{equation}
	F_K(x,a,q) = \left. P_{Q_{S^3\backslash K}}(\boldsymbol{x},q) \right|_{x_i=x^{n_i} a^{a_i} q^{q_i}}	
	\end{equation}
	where $n_i,a_i,q_i$ are fixed integers.

One of interesting implications of above conjecture is a~definition of analogues of LMOV invariants for knots complements based on DT invariants of corresponding quivers \cite{Kuch2005}:
\begin{equation}
    N_{S^3\backslash K}(x,a,q)=\left.\Omega_{Q_{S^3\backslash K}}(\boldsymbol{x},q)\right|_{x_i=xa^{a_i}q^{q_i}}
\end{equation}
In Section \ref{sec:knot_complements} we use Theorem \ref{thm:BPS_spectrum_of_augmented_quiver} to describe the whole infinite set of  analogues of LMOV invariants for complements of $(2,2p+1)$ torus knots.

\subsection{Block decomposition of quivers via HOMFLY-PT skein}

In order to establish the relation between twisting the knot and unlinking the quiver, we first need to know how similar are quivers
corresponding to knots or links obtained from crossing resolutions of a~given knot or a~link. In other words, we need to
learn a~possible generalisation of the HOMFLY-PT skein relation for symmetric quivers. This is a~challenging problem to which we are not yet ready to give a~complete solution. However, what we discuss below is an important step in this direction, and it will be sufficient for study of the twist dependence.

Our starting observation is the following.
Let $Q$ be the quiver of a~knot $K$ such that its two diagonal blocks are the two quivers $Q'$ and $Q''$, corresponding to knots $K'$ and $K''$:
\begin{equation}
Q=\left[\begin{array}{c:c}
Q'& \ast\\
\hdashline
\ast & Q''
\end{array}
\right].
\end{equation}
Then we have:
\begin{equation}\label{eq:K_1 + K_2}
P_1(K)=P_1(K')+P_1(K'').
\end{equation}
Also then the number of summands in each $P_1(*)$ is equal to the number of corresponding indices in the index set of $Q$.  Denote the set of indices corresponding to $Q'$ by $I$ and the set of indices corresponding to $Q''$ by $J$. Then
$$P_1(K')=\sum_{i\in I} (-1)^{s'_i} a^{a'_i} q^{q'_i}$$
$$P_1(K'')=\sum_{j\in J} (-1)^{s''_j} a^{a''_j} q^{q''_j}.$$

For the second symmetric polynomial we then have:
$$P_2(K)=P_2(K')+P_2(K'')+(1+q^2)\sum_{i\in I}\sum_{j\in J} (-1)^{s'_i+s''_j} a^{a'_i+a''_j} q^{q'_i+q''_j}\times q^{C_{ij}},$$
where $C_{ij}$ is the $(i,j)$ entry of the quiver matrix $Q$, which is precisely the off-diagonal block.
Therefore we have
\begin{equation}\label{f3}
P_2(K)-P_2(K')-P_2(K'')=(1+q^2)\sum_{i\in I}\sum_{j\in J} (-1)^{s'_i+s''_j} a^{a'_i+a''_j} q^{q'_i+q''_j}\times q^{C_{ij}},
\end{equation}
and so the left hand side resemble the terms in the two brackets on the LHS of (\ref{f2}), whereas the RHS 
of (\ref{f3}) has the factor $1+q^2$ (as the RHS of (\ref{f2})), as well as some ``twisted" product of the uncoloured polynomials of $P_1(K_1)$ and $P_1(K_2)$, where the twist is determined by the off-diagonal block in $Q$.
Similar considerations for a~quiver $Q$ of the given knot $K$ describes how the entries of $C_{ij}$ ``relate" $P_2(K)$ and $(P_1(K))^2$.

On the other hand, let us take a~closer look at HOMFLY-PT skein relation applied to a~link $L$, which we write in the following form:
\begin{equation}\label{eq:homfly skein in a subquiver form}
    P(L_+)=a^{-2}P(L_-)+a^{-1}(q^{-1}-q)P(L_0).
\end{equation}
Assume further that both $L_{+}$ and $L_{-}$ are knots, while $L_0$ is a~two-component link. We can identify the above equation with (\ref{eq:K_1 + K_2}), by noticing that multiplication by $a^{-2}$ can be interpreted as an effect of framing $L_-$,\footnote{HOMFLY-PT invariant for an $f$-framed knot is defined as $a^{2f}q^{2fr(r-1)}P_r(a,q)$, where $P_r(a,q)$ is 0-framed invariant.}
while $a^{-1}(q^{-1}-q)P(L_0)$ is a~finite link polynomial corresponding to $L_0$.\footnote{A version of knot-quiver correspondence for such finite link polynomials was considered in \cite{KRSS1707long}.}
The above said suggests that quiver $Q_{L_{+}}$ admits a~block form decomposition, where the two blocks correspond to modifications of a~crossing in $L_+$:
\begin{equation}\label{eq:block decomposition of L_+ quiver}
Q_{L_+}=\left[\begin{array}{c:c}
Q_{L_-}^{f=-1}& \ast\\
\hdashline
\ast & Q_{L_0}
\end{array}
\right],
\end{equation}
where $Q_{L_-}^{f=-1}$ denotes a~quiver for $L_-$ with framing $-1$, and $Q_{L_0}$ is a~quiver corresponding to the two-component link, and is understood in the same sense as in \cite{KRSS1707long}. The relation (\ref{eq:block decomposition of L_+ quiver}) is very non-trivial, as it captures the behaviour of all $Sym^r$-coloured HOMFLY-PT polynomials under skeining at once\footnote{Even more challenging seems the identification of the off-diagonal terms in (\ref{eq:block decomposition of L_+ quiver}).}.
(Analogous relation can be obtained for the negative crossing in $K$, if we put $L_-$ on the left and divide everything by $-a^{-2}$.)
Note that for a~knot $K$, $Q_K$ can admit many different block forms (\ref{eq:block decomposition of L_+ quiver}), given by modifications of different crossings in $K$. 
Summing up, we suggest the following relation between HOMFLY-PT skeins and symmetric quivers:

\begin{conjecture}\label{coj:skein relation for quivers}
Given a~knot $K$ and its quiver $Q_K$, we can form a~set of knots $\{K_{i,\sigma}\}$ for every $i$-th crossing of $K$ of type $L_+$, obtained by modifying this crossing: $\sigma\in\{-,0\}$. Then the HOMFLY-PT skein relation (\ref{eq:homfly skein in a subquiver form}) can be lifted to the following statement: for every resolution $i,\sigma$, $Q_{K_{i,\sigma}}$ is a~subquiver in $Q_K$ (in some appropriate framing). Moreover, $Q_K$ admits a~block form composed from any two such resolutions.
\end{conjecture}

\noindent Below we we confirm this conjecture for the figure-eight knot which has two equivalent symmetric quivers $Q_{4_1}$ and $Q'_{4_1}$ \cite[Section 5.2]{JKLNS2105}.
    
Figure-eight knot has 4 crossings, and flipping either of them gives the unknot, while smoothing either of them gives the mirror of Hopf link (Figure \ref{fig:figure-eight knot}).
    \begin{figure}[ht!]
        \centering
        \includegraphics[width=0.25\linewidth]{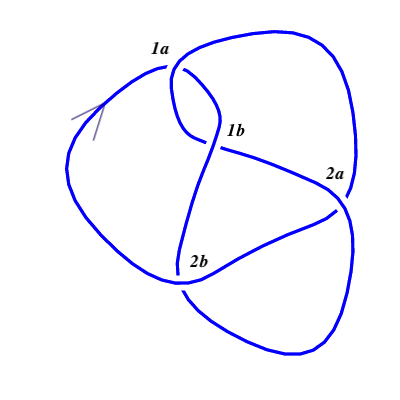}
        \caption{Knot $4_1$.}
        \label{fig:figure-eight knot}
    \end{figure}
However, we need to be careful with the orientation of strands, as we may (at least in principle) get inequivalent recursions when applying the skein relation to different crossings. Skein relation applied to either (1a) or (1b) which are of type $L_{-}$, gives
\begin{equation}\label{eq:first recursion fig-8}
P(4_1)-a^{-2}P(0_1)=a^{-1}(q^{-1}-q)P(m\mathrm{L2a1}).
\end{equation}
On the other hand, we can flip all crossings on the above diagram simultaneously without changing the knot but changing the orientation of the Hopf link, which gives
    \begin{equation}\label{eq:mirror recursion fig-8}
    P(4_1)-a^{2}P(0_1)=a(q-q^{-1})P(\mathrm{L2a1}).
    \end{equation}
    Both relations can be seen diagrammatically, if we plot the exponents $(a_i,q_i)$ entering the uncoloured HOMFLY-PT polynomial of $4_1$:
    \[\begin{tikzcd}
	& {{\bullet}} &&& {\textcolor{blue}{\bullet}} \\
	{{\bullet}} & \bullet & \bullet & {\textcolor{blue}{\bullet}} & {\textcolor{blue}{\bullet}} & {\textcolor{blue}{\bullet}} \\
	& {\textcolor{blue}{\bullet}} &&& \bullet \\
	& {\textcolor{blue}{\bullet}} &&& \bullet \\
	\bullet & \bullet & \bullet & {\textcolor{blue}{\bullet}} & {\textcolor{blue}{\bullet}} & {\textcolor{blue}{\bullet}} \\
	& \bullet &&& {\textcolor{blue}{\bullet}}
	\arrow[from=1-2, to=2-1]
	\arrow[from=1-2, to=2-3]
	\arrow[from=1-5, to=2-4]
	\arrow[from=1-5, to=2-6]
	\arrow[from=2-1, to=3-2]
	\arrow[from=2-3, to=3-2]
	\arrow[from=2-4, to=3-5]
	\arrow[from=2-6, to=3-5]
	\arrow[from=4-2, to=5-1]
	\arrow[from=4-2, to=5-3]
	\arrow[from=4-5, to=5-4]
	\arrow[from=4-5, to=5-6]
	\arrow[from=5-1, to=6-2]
	\arrow[from=5-3, to=6-2]
	\arrow[from=5-4, to=6-5]
	\arrow[from=5-6, to=6-5]
\end{tikzcd}\]
This suggests that both the shifted unknot (the blue node on the top-left and bottom-left pictures) and Hopf link or its mirror (4 blue nodes on the right pictures) are subquivers in either $Q_{4_1}$ or $Q'_{4_1}$ which can be checked directly:
\begin{equation}\label{eq:Qa for 4_1}
\renewcommand\arraystretch{1.2
}
Q_{4_1}=\left[
\begin{array}{cccc:c}
 0 & -1 & -1 & 0 & 0 \\
 -1 & -2 & -2 & -1 & 0 \\
 -1 & -2 & -1 & -1 & 0 \\
 0 & -1 & -1 & 1 & 1 \\
 \hdashline`
 0 & 0 & 0 & 1 & 2 \\
\end{array}
\right],\ 
\left[
\begin{array}{c}
 1 \\
 a^{-2} q^2 \\
 q^{-1} \\
 q \\
 \hdashline
 {a^2 q^{-2}} \\
\end{array}
\right].
\end{equation}
Detaching the unknot in framing $+2$ (the last node), we get a~$4\times 4$ matrix and a~vector for the Hopf link, matching exactly those from \cite{KRSS1707long}:
\begin{equation}\renewcommand\arraystretch{1.2
}
    \left[
\begin{array}{cccc}
 0 & -1 & -1 & 0 \\
 -1 & -2 & -2 & -1 \\
 -1 & -2 & -1 & -1 \\
 0 & -1 & -1 & 1 \\
\end{array}
\right],\ 
\left[
\begin{array}{c}
 1 \\
 a^{-2} q^2 \\
 q^{-1} \\
 q \\
\end{array}
\right]
\ \xrightarrow[]{\text{\emph{framing} }+2} \ 
\left[
\begin{array}{cccc}
 2 & 1 & 1 & 2 \\
 1 & 0 & 0 & 1 \\
 1 & 0 & 1 & 1 \\
 2 & 1 & 1 & 3 \\
\end{array}
\right],\ 
\left[
\begin{array}{c}
 1 \\
 a^{-2} q^2 \\
 q^{-1} \\
 q \\
\end{array}
\right]
\times a^2 q^{-2}.
\end{equation}
In order to see the mirror of Hopf link, we can swap the nodes 1 and 2 in (\ref{eq:Qa for 4_1}):
\begin{equation}\renewcommand\arraystretch{1.2
}
   Q_{4_1} =  \left[
\begin{array}{c:cccc}
 -2 & -1 & -2 & -1 & 0 \\
 \hdashline
 -1 & 0 & -1 & 0 & 0 \\
 -2 & -1 & -1 & -1 & 0 \\
 -1 & 0 & -1 & 1 & 1 \\
 0 & 0 & 0 & 1 & 2 \\
\end{array}
\right],\ 
\left[
\begin{array}{c}
 a^{-2} q^2 \\
 \hdashline
 1 \\
 q^{-1} \\
 q \\
 a^2 q^{-2} \\
\end{array}
\right].
\end{equation}
Mirroring operation is given by $(C,\boldsymbol{x}) \mapsto (I-C-[1],\boldsymbol{x}|_{a\mapsto a^{-1},q\mapsto q^{-1}})$. We can see that the two subquivers are almost mirror of each other, apart from one pair of arrow.
\begin{equation}\renewcommand\arraystretch{1.2
}
    \left[
\begin{array}{cccc}
 0 & -1 & -1 & 0 \\
 -1 & -2 & -2 & -1 \\
 -1 & -2 & -1 & -1 \\
 0 & -1 & -1 & 1 \\
\end{array}
\right],\ 
\left[
\begin{array}{c}
 1 \\
 a^{-2} q^2 \\
 q^{-1} \\
 q \\
\end{array}
\right] \quad
\leftrightarrow
\quad
\left[
\begin{array}{cccc}
 0 & -1 & 0 & 0 \\
 -1 & -1 & -1 & 0 \\
 0 & -1 & 1 & 1 \\
 0 & 0 & 1 & 2 \\
\end{array}
\right],\ 
\left[
\begin{array}{c}
 1 \\
 q^{-1} \\
 q \\
 a^{2} q^{-2} \\
\end{array}
\right].
\end{equation}
This confirms our proposition that $Q_{4_1}$ admits block decomposition (\ref{eq:block decomposition of L_+ quiver}), which takes form
\begin{equation}
  \left[\begin{array}{ c : c }
    Q^{f=1}_{\bigcirc} & * \\
    \hdashline
    * & Q_{\text{Hopf}}
  \end{array}\right]
.
\end{equation}
Note that $Q_{\bigcirc}$ acquires framing shift $+1$, reflecting the presence of the factor $a^2$ in (\ref{eq:first recursion fig-8}) and (\ref{eq:mirror recursion fig-8}). On the other hand, the factor $a^{-1}(q-q^{-1})$ turns $P(\mathrm{L2a1})$ into polynomial form, compatible with the finite-dimensional version of torus link homology (see \cite{KRSS1707long} for the details). Next, we consider equivalent quiver~$Q'_{4_1}$:
\begin{equation}\renewcommand\arraystretch{1.2
}
Q'_{4_1}=\left[
\begin{array}{cccc:c}
 0 & -1 & -1 & 0 & 0 \\
 -1 & -2 & -2 & -1 & -1 \\
 -1 & -2 & -1 & 0 & 0 \\
 0 & -1 & 0 & 1 & 1 \\
 \hdashline`
 0 & -1 & 0 & 1 & 2 \\
\end{array}
\right],\ 
\left[
\begin{array}{c}
 1 \\
 a^{-2} q^2 \\
 q^{-1} \\
 q \\
 \hdashline
 {a^2 q^{-2}} \\
\end{array}
\right].
\end{equation}
This gives
\begin{equation}\renewcommand\arraystretch{1.2
}
  C'_{4_1} = 
  \left[\begin{array}{ c : c }
    Q^{f=-1}_{\bigcirc} & * \\
    \hdashline
    * & Q_{m\text{Hopf}}
  \end{array}\right]
  \neq
    \left[\begin{array}{ c : c }
    Q^{f=1}_{\bigcirc} & * \\
    \hdashline
    * & Q_{\text{Hopf}}
  \end{array}\right]\,.
\end{equation}
Note that the two block decompositions correspond to two different, but equivalent quivers. However, since $Q_{4_1}$ and $Q'_{4_1}$ can be unlinked to the same quiver, the two block decompositions will be combined in this unlinked quiver.

We can also check other crossings. In another twist region, crossings $2a$ and $2b$ in Figure \ref{fig:figure-eight knot} are of type $L_{+}$. The skein relation applied to either of them gives
    \begin{equation}\label{eq:fig-8 skeining tw2}
    P(4_1)-a^2P(0_1)=a(q-q^{-1})P(\rm mL2a1\{1\}),
    \end{equation}
    where $\rm L2a1\{1\}$ differs from Hopf link by orientation reversal of one of its components. Note, however, that $P({\rm mL2a1\{1\}})\equiv P(\mathrm{L2a1})$, thus on the polynomial level we get the same relation as above.

\begin{rmk}\label{rmk:block quiver conjecture -- iterated}
The above operations can be iterated several times, until all crossings of a~given knot are completely resolved and we get a~collection of unknots. In the above example,
we can start with
\begin{equation}\renewcommand\arraystretch{1.2
}
      Q^a_{4_1} = 
  \left[\begin{array}{ c : c }
    Q^{f=1}_{\bigcirc} & * \\
    \hdashline
    * & Q_{\text{\rm Hopf}}
    \end{array}
    \right].
\end{equation}
Applying HOMFLY-PT skein once more to the Hopf link, we get
\begin{equation}\label{eq:4_1 quiver block decomposition full}\renewcommand\arraystretch{1.2
}
      Q^a_{4_1} = 
  \left[\begin{array}{ c : c : c }
    Q^{f=1}_{\bigcirc} & * & * \\
    \hdashline
    * & Q^{f=0}_{\bigcirc} & * \\
    \hdashline
    * & * & R'
    \end{array}
    \right] = 
    \left[\begin{array}{ c : c }
    Q^{f=0}_{\bigcirc} & * \\
    \hdashline
    * & R
    \end{array}
    \right],
\end{equation}
where $R'$ corresponds to a~collection of unknots coming from skeining the Hopf link (they not only come in different framings but also have some $q$-shifts), while $R$ combines $Q^{f=1}_{\bigcirc}$ with $R'$. The same steps can be applied, in principle, to any knot -- it would be very interesting to make this more precise and, if possible, to lift this statement to the level of HOMFLY-PT homology.
\end{rmk}

\subsection{Twisting and unlinking: knot-quiver stable growth conjecture} 

Structural properties of $Sym^r$-coloured HOMFLY-PT invariants of rational links allow the HOMFLY-PT generating series $P_{K_i}(x,a,q)$ to be written in a~form of
some operator acting on $P_{K_1}(x,a,q)$.
In particular, for any two knots which differ by a~twist, the twist dependence can be effectively absorbed into such an operator \cite{SW1711,SW2004}. In a~similar context \cite{JKLNS2105}, certain $q$-hypergeometric identities were a~motivation to introduce the $(k,l)$-splitting operation on quivers.
Our new findings suggest that operations of twisting the knot and splitting the quiver are in fact very closely related. 

In order to formulate the precise statement, we first note that Conjecture \ref{coj:skein relation for quivers} implies that $Q_{K_1}$ admits a~block form decomposition which is a~generalisation of (\ref{eq:4_1 quiver block decomposition full}):\footnote{We always assume framing zero, i.e. $Q_{K_1}=Q^{f=0}_{K_1}$.}
\begin{equation}\label{eq:K_1 block form with K_0}\renewcommand\arraystretch{1.2
}
Q_{K_1}=\left[\begin{array}{c:c}
Q^{f=0}_{K_0}& \ast\\
\hdashline
\ast & R
\end{array}
\right].
\end{equation}
In this block decomposition, $R$ does not necessarily correspond to any link. For example, it can be obtained by iterating (\ref{eq:block decomposition of L_+ quiver}) several times, which may result in mixing between the blocks corresponding to links obtained along the way, see Remark \ref{rmk:block quiver conjecture -- iterated}.
Since by definition of $K_1$, removing a~full twist produces $K_0$, we can remove a~full twist simply by switching the type of one crossing and then applying a~Reidemeister move. It means that we can use skein relation to relate $K_1$, $K_0$ and some link which is the result of smoothing the corresponding crossing. According to Conjecture \ref{coj:skein relation for quivers}, at the level of quivers this corresponds to a~block form decomposition as above.
In consequence, our main conjecture can be stated as follows:

\begin{conjecture}[Knot-quiver stable growth conjecture]
\label{coj:knot quivers twists}
Let $\{K_i\}_{i\in\mathbb{Z}_+}^\tau$ be a~family of knots generated by $K_1$ by adding full twists. Let also $Q_1:=Q_{K_1}$ be a~quiver corresponding to the knot $K_1$.
%
%
Then, performing full twists on a~knot translates into the following operations on the corresponding quiver:
\begin{enumerate}
    \item Find a~correct augmentation $Q^+_1$, which requires two substeps.
    \begin{itemize}
    \item First, decompose $Q_1$ into the block form (\ref{eq:K_1 block form with K_0}).
    \item Second, augment $Q_1$ in the following way. 
    \begin{itemize}
    \item If the strands in the twist region have the opposite orientation, we choose the augmentation which corresponds to $(0,2)$-splitting of all nodes which are not part of $Q^{f=0}_{K_0}$:
\begin{equation}\renewcommand\arraystretch{1.2
}\label{eq:unlink augmentation}
Q^+_{1} = 
\left[
\begin{array}{c|c:c}
 1 & 0 & 1 \\
 \hline
 0 & Q^{f=0}_{K_0} & * \\
 \hdashline
 1 & * & R  \\
\end{array}
\right],
\{x_0=a^2q^{-1}\} \cup \boldsymbol{x}.
\end{equation}
\item If the strands in the twist region have the same orientation, we choose the augmentation which corresponds to $(1,2)$-splitting of all nodes which are not part of $Q^{f=0}_{K_0}$:
\begin{equation}\renewcommand\arraystretch{1.2
}\label{eq:link augmentation}
Q^+_{1} = 
\left[
\begin{array}{c|c:c}
 0 & 2 & 1 \\
 \hline
 2 & Q^{f=0}_{K_0} & * \\
 \hdashline
 1 & * & R  \\
\end{array}
\right],
\{x_0=q^2\} \cup \boldsymbol{x}.\hphantom{q^{-1}}
\end{equation}
\end{itemize}
\end{itemize}
    \item For $i=1,2,\dots$, define $Q^+_{i+1}$ by unlinking (or linking) all units $C_{0\iota}=1$ in case of $(0,2)$-splitting (or $(1,2)$-splitting) in $Q^+_i$:
    \begin{equation}
    Q^+_{i+1} := U(0\iota_1)\dots U(0\iota_{|R|})Q^+_i,\quad \text{or}\quad
    Q^+_{i+1} := L(0\iota_1)\dots L(0\iota_{|R|})Q^+_i.
    \end{equation}
    Note that if $i=1$, the arrows which we (un)link in $Q_1^+$, connect the extra node with $R$. By iterating this procedure, we simply unlink or link the new nodes created after $R$. After removing the extra node, we have $Q_i=Q_{K_i}$.
\end{enumerate}
 Summing up, the quivers $Q_i$ obtained by unlinking $Q^+_{1}$ and removing the extra node correspond to knots $K_i$ obtained by adding full twists to $K_1$.
\end{conjecture}

\begin{rmk}
    Analogous conjecture can be proposed for Poincaré polynomials -- indeed, at the quiver level, at least for simple knots, it is almost effortless by introducing the $t$-grading via $x_i=a^{a_i}q^{q_i}(-t)^{C_{ii}}x$. However, in this case the reasoning behind the augmented quiver would be changed, because skein relation in the usual form does not hold for superpolynomials. It is therefore an interesting question if there is a~good homological formulation of this conjecture.
\end{rmk}

Note that the above conjecture implies stability for coloured HOMFLY-PT poynomials discussed in Section \ref{sec:Stable growth for knots}. If $Q_{i+1}$ can be obtained
from $Q_{i}$ by $(k,l)$-splitting for every $i\in\mathbb{Z}_{+}$,
then, by definition of splitting, $Q_i$ is a~subquiver in $Q_{i+1}$. As a~result, we get $\mathcal{P}(K_{i+1};a,q)=\mathcal{P}(K_{i};a,q)+(a^{a_{s}}q^{q_{s}}(-1)^{t_{s}})^{i}\mathcal{S}(K_0;a,q)$,
where $a_{s},q_{s}\in\mathbb{Z}$ and $\mathcal{S}(K_0;a,q)$
is a~polynomial independent from $i$.

\begin{rmk}
    We confirmed Conjecture \ref{coj:knot quivers twists} for all pretzel knots up to 15 crossings with an odd number of crossings in each twist region, which require the application of $(0,2)$-splitting.\footnote{The data of the resulting quiver matrices is available upon request.}
\end{rmk}

\subsection{Interpretation of the augmented quiver}\label{sec:Interpretation of an augmented quiver}

Below we provide an interpretation of the quiver generating series of the augmented quiver $Q^+_K$ in terms of link invariants. Let $K=K_1$ be coming from some family $\{K_i\}^\tau_{i\in \mathbb{Z}}$ parametrised by full twists in $\tau$, as in Definition \ref{defn:twisted family}. Following Conjecture \ref{coj:knot quivers twists}, we consider the two cases:
\begin{enumerate}
    \item Anti-parallel orientation of strands in $\tau$
    \item Parallel orientation of strands in $\tau$
\end{enumerate}
In the first case, we can start with $Q^+_K=Q_1^+$ given by (\ref{eq:unlink augmentation}) and apply unlinking infinitely many times to obtain $Q^+_{\infty}$:
    \begin{equation}\renewcommand\arraystretch{1.5
}
    Q^+_{1} = 
    \left[
    \begin{array}{c|c:c}
    1 & 0 & 1 \\
    \hline
    0 & Q^{f=0}_{K_0} & * \\
    \hdashline
    1 & * & R  \\
    \end{array}
    \right]
    \quad
    \longrightarrow
    \quad
    Q^+_{\infty} = 
    \left[
    \begin{array}{c|c:c:c}
    1 & 0 & 0 & \cdots \\
    \hline
    0 & Q^{f=0}_{K_0} & * & \vdots \\
    \hdashline
    0 & * & R & \vdots  \\
    \hdashline
    \vdots & \cdots & \cdots & \ddots
    \end{array}
    \right]
    \end{equation}
    
     \noindent Note that its quiver generating series is stable, as confirmed in Section \ref{sec:Quiver perspective}. Moreover, the extra node is  completely disconnected in $Q^+_{\infty}$, while the rest of the quiver corresponds to the infinitely twisted knot $K_{\infty}$. In order to make sense of it, we can use the stable results from Section~ \ref{sec:Limits for infinite twists} (in particular, Proposition \ref{prp:infinite tw knot into a link}) combined with the knot-quiver correspondence. We deduce that $Q^+_{\infty}$ encodes HOMFLY-PT invariants of a~link obtained by smoothing $\tau$, but so does $Q_1^+$, since they are equivalent via unlinking and have the same quiver generating series. Therefore, we find that \emph{for the anti-parallel orientation of strands, augmented quiver $Q^+_K$ encodes link invariants of a~link obtained by applying 0-resolution to the corresponding twist region}. Of course, this holds up to a~multiplication by the $q$-Pochhammer from the disconnected extra node, which, however, does not depend on $x$ and can be thought of as a~normalisation prefactor. We will illustrate this in Section \ref{sec:Case studies}.

In the second case, we can start with $Q_1^+$ given by (\ref{eq:link augmentation}) and apply linking infinitely many times to obtain $Q^+_{\infty}$:
    \begin{equation}\renewcommand\arraystretch{1.5
}
    Q^+_{1} = 
\left[
\begin{array}{c|c:c}
 0 & 2 & 1 \\
 \hline
 2 & Q^{f=0}_{K_0} & * \\
 \hdashline
 1 & * & R  \\
\end{array}
\right]
    \quad
    \longrightarrow
    \quad
    Q^+_{\infty} = 
    \left[
    \begin{array}{c|c:c:c}
    0 & 2 & 2 & \cdots \\
    \hline
    2 & Q^{f=0}_{K_0} & * & \vdots \\
    \hdashline
    2 & * & R & \vdots  \\
    \hdashline
    \vdots & \cdots & \cdots & \ddots
    \end{array}
    \right]
    \end{equation}
Unlike in the first case, the extra node is not disconnected in $Q^+_{\infty}$, and we have to be careful with it. (One possible solution would be to choose the overall framing such that it becomes disconnected.) This case is more subtle, and we leave its general interpretation as an open question. The case of torus knots, however, is very special since the exact limit is available and well-studied -- we will consider it in Section \ref{sec:Case studies} as well.

\section{Case studies}\label{sec:Case studies}

In this Section we study several examples of knot families $\{K_i\}_{i\in\mathbb{Z}_+}^\tau$, all of which confirm Conjecture~ \ref{coj:knot quivers twists}. Apart from them, we provide some evidence for knot complement invariants -- the case not covered by the main conjecture. As we found experimentally, the quivers for torus knot complements require the application $(-1,-2)$-splitting. Lastly, we study the stable behaviour of LMOV invariants of knots and their complements, illustrating Theorem \ref{thm:BPS_spectrum_of_augmented_quiver} in the knot-quiver setting.

\subsection{Twist knots: \texorpdfstring{$K_p$}{Kp}}\label{sec:Twist_knots}

\subsubsection{\texorpdfstring{$p<0: 4_{1}$, $6_{1}$, $8_{1}$}{41, 61, 81}, etc.} \label{sec:41,61,81}

We start with $K_1=4_1$. The figure-eight knot has two twist regions, both with anti-parallel orientation of strands, which suggests the application of (\ref{eq:unlink augmentation}). Twist knot family is generated by adding full twists in the gray region (Figure \ref{fig:4_1 with tau2}).

\begin{figure}[ht!]
    \centering
    \includegraphics[width=0.2\linewidth]{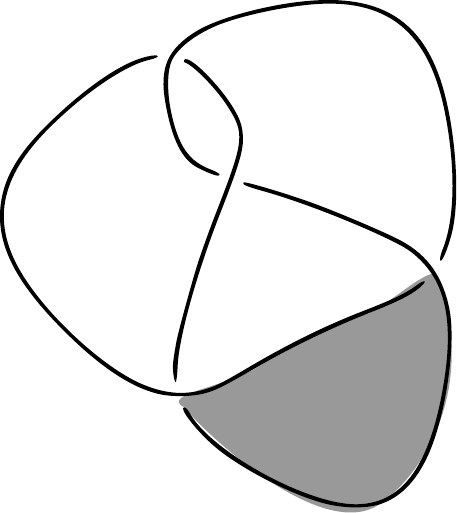}
    \caption{Knot $4_1$. The relevant twist region $\tau$ shown in gray.}
    \label{fig:4_1 with tau2}
\end{figure}

\noindent Consider a~quiver $Q_1=Q_{4_1}$ which was found in \cite{KRSS1707long}:\footnote{We could have chosen another, equivalent quiver corresponding to the $4_1$ knot, given in \cite{EKL1910,JKLNS2105}.}
\begin{equation}
\renewcommand\arraystretch{1.5
}
Q_{4_1} = 
\left[
\begin{array}{ccccc}
 0 & -1 & -1 & 0 & 0 \\
 -1 & -2 & -2 & -1 & 0 \\
 -1 & -2 & -1 & -1 & 0 \\
 0 & -1 & -1 & 1 & 1 \\
 0 & 0 & 0 & 1 & 2 \\
\end{array}
\right],
\left[
\begin{array}{c}
 x \\
 a^{-2} q^{2} x \\
 q^{-1} x \\
 q x \\
 a^2 q^{-2} x\\
\end{array}
\right].
\end{equation}
In order to find the correct augmentation $Q^+_{1}$, note that removing a~full twist from $4_1$ gives the~ unknot $K_0=\bigcirc$.
The corresponding subquiver $Q^{f=0}_{\bigcirc}$ can be found from the block decomposition (\ref{eq:K_1 block form with K_0}). As a~result, the augmented quiver is given by (see also Figure \ref{fig:augmentation_unlinking}, left)
\begin{equation}\renewcommand\arraystretch{1.5
}\label{eq:4_1 augmented quiver}
Q^+_{1}= 
\left[
\begin{array}{c|c:cccc}
 1 & 0 & 1 & 1 & 1 & 1 \\
 \hline
 0 & 0 & -1 & -1 & 0 & 0 \\
 \hdashline
 1 & -1 & -2 & -2 & -1 & 0 \\
 1 & -1 & -2 & -1 & -1 & 0 \\
 1 & 0 & -1 & -1 & 1 & 1 \\
 1 & 0 & 0 & 0 & 1 & 2 \\
\end{array}
\right],
\left[
\begin{array}{c}
  a^2 q^{-1} \\
  \hline
 x \\
 \hdashline
 a^{-2} q^{2} x \\
 q^{-1} x \\
 q x \\
 a^2 q^{-2} x\\
\end{array}
\right].
\end{equation}
(Here the block decomposition (\ref{eq:K_1 block form with K_0}) is shown by dashed lines.)
\noindent Using this, we can check if twisting the knot indeed corresponds to unlinking the quiver. We model $(0,2)$-splitting by unlinking and expect that from
\begin{equation}\label{eq:negative_twist_knots_unlinkings}
    Q^{+}_{i} = U(0,3i-1)U(0,3i-2)U(0,3i-3)U(0,3i-4)\,Q^+_{i-1},\quad i=2,3,\dots
\end{equation}
we can identify $Q_i\equiv Q_{K_i}$.
For $i=1$ we get
\begin{align}\label{eq:6_1 augmented quiver}
    Q^{+}_2  & = U(0,5)U(0,4)U(0,3)U(0,2)\,Q^+_1 \nonumber  \\
&=
\renewcommand\arraystretch{1.2}
\left[
\begin{array}{c|c:cccc:cccc}
 1 & 0 & 0 & 0 & 0 & 0 & 1 & 1 & 1 & 1 \\
 \hline
 0 & 0 & -1 & -1 & 0 & 0 & -1 & -1 & 0 & 0 \\
 \hdashline
 0 & -1 & -2 & -2 & -1 & 0 & -2 & -2 & -1 & 0 \\
 0 & -1 & -2 & -1 & -1 & 0 & -1 & -1 & -1 & 0 \\
 0 & 0 & -1 & -1 & 1 & 1 & 0 & 0 & 1 & 1 \\
 0 & 0 & 0 & 0 & 1 & 2 & 1 & 1 & 2 & 2 \\
 \hdashline
 1 & -1 & -2 & -1 & 0 & 1 & 0 & 0 & 1 & 2 \\
 1 & -1 & -2 & -1 & 0 & 1 & 0 & 1 & 1 & 2 \\
 1 & 0 & -1 & -1 & 1 & 2 & 1 & 1 & 3 & 3 \\
 1 & 0 & 0 & 0 & 1 & 2 & 2 & 2 & 3 & 4 \\
\end{array}
\right],
\left[
\begin{array}{c}
 {a^2 q^{-1}} \\
 \hline
 x \\
 \hdashline
 a^{-2} q^{2} x \\
 q^{-1}x \\
 q x\\
 a^2 q^{-2} x\\
 \hdashline
 x \\
 a^2 q^{-3} x \\
 a^2 q^{-1} x \\
 a^4 q^{-4} x \\
\end{array}
\right]\;.
\end{align}
Indeed, after removing the extra node, $Q_2$ agrees with a~quiver $Q_{K_2}$ corresponding to $6_1$ knot \cite{KRSS1707long}.
Comparing \eqref{eq:negative_twist_knots_unlinkings} with the recursive construction of quivers corresponding to twist knots $4_{1}$, $6_{1}$, $8_{1}$, etc. in \cite{KRSS1707long}, we can see that for this infinite family Conjecture \ref{coj:knot quivers twists} is true.

Let us move our attention to BPS states and see how Theorem \ref{thm:BPS_spectrum_of_augmented_quiver} enables encoding all of them in a~finite augmented quiver. We know that $\Theta_{Q^+_{4_1}}=\Theta_{Q^+_{6_1}}=\dots=\Theta_{Q^+_{\infty}}$, which can be seen at the level of their counts given by analogues of LMOV invariants coming from (\ref{eq:4_1 augmented quiver}) and (\ref{eq:6_1 augmented quiver}):
\[
    N_{4_1^+}(x,a,q)=N_{6_1^+}(x,a,q)=a^2 + \left(-\frac{1}{a^2}+\left(-2+\frac{1}{q^2}+q^2\right)+\left(-2+\frac{1}{q^2}+q^2\right)a^2+\mathcal{O}(a^4)\right) x+\mathcal{O}(x^2)\,.
\]
Let us stress that in contrary to the standard LMOV invariants, there is an infinite number of terms for a~given power of $x$, which is the reason why we cut off at fixed power of $a$. This indicates that the augmented quiver in fact encodes link invariants, see Section \ref{sec:Interpretation of an augmented quiver} (we will also discuss it in more detail below).
We can obtain $N_{4_1}(x,a,q)$ by setting $x_0=0$ instead of $x_0=a^2q^{-1}$ in the knot-quiver change of variables in (\ref{eq:4_1 augmented quiver}), which leads to
\begin{equation}\label{eq:LMOV_4_1}
    N_{4_1}(x,a,q)=\left(-\frac{1}{a^2}+\left(-1+\frac{1}{q^2}+q^2\right)-a^2\right) x+\mathcal{O}(x^2)\,.
\end{equation}
$N_{6_1}(x,a,q)$ arises from setting $x_0=0$ instead of $x_0=a^2q^{-1}$ in (\ref{eq:6_1 augmented quiver}), which excludes less states from $\Theta_{Q^+_{4_1}}=\Theta_{Q^+_{6_1}}=\dots=\Theta_{Q^+_{\infty}}$:
\begin{equation}\label{eq:LMOV_6_1}
    N_{6_1}(x,a,q)= \left(-\frac{1}{a^2}+\left(-2+\frac{1}{q^2}+q^2\right)+\left(-1+\frac{1}{q^2}+q^2\right)a^2-a^4\right) x+\mathcal{O}(x^2)\,.
\end{equation}
Analogous computations can be done up to infinity and in the limit we have $\Theta_{Q_{\infty}}$ which is equal to $\Theta_{Q^+_{4_1}}$ minus a~single state coming from the extra node, so the corresponding count of BPS states is given by 
\begin{equation}\label{eq:LMOV_infty}
    N_{\infty}(x,a,q)=\left(-\frac{1}{a^2}+\left(-2+\frac{1}{q^2}+q^2\right)+\left(-2+\frac{1}{q^2}+q^2\right)a^2+\mathcal{O}(a^4)\right) x+\mathcal{O}(x^2)\,.
\end{equation}
Comparing equations (\ref{eq:LMOV_4_1}-\ref{eq:LMOV_infty}) we can see that indeed $\Theta_{Q_{4_1}}\subset\Theta_{Q_{6_1}}\subset\dots\subset\Theta_{Q_{\infty}}$.

\paragraph{} Lastly, we illustrate the interpretation of an augmented quiver following Section \ref{sec:Interpretation of an augmented quiver}.
To this end, consider augmented quiver $Q_1^+$ given by \eqref{eq:4_1 augmented quiver}.
Thanks to unlinking, it is equivalent to a~quiver corresponding to the infinitely twisted knot $K_{\infty}$, augmented with an extra node.
According to the interpretation in Section \ref{sec:Interpretation of an augmented quiver}, its quiver partition function agrees with the generating series of HOMFLY-PT polynomials of the Hopf link, after dividing by $q$-Pochhammer $(a^2;q^2)_{\infty}$. Let's take a look at the linear term in $x$ in greater detail.
The contributions of $x_i$ corresponding to the nodes which are to be unlinked in $Q_1^+$ ($i=2,\dots,5$), are
\begin{equation}
\begin{aligned}
x_i\sum_{d_1\geq 0}\frac{(-q)^{d_1^2+2d_1+C_{ii}}}{(q^2;q^2)_{d_1}}\left(a^2q^{-1}\right)^{d_1} = &\ x_i(-q)^{C_{ii}}(a^2q^2;q^2)_{\infty},
\end{aligned}
\end{equation}
so that dividing it by $(a^2;q^2)_{\infty}$ gives $x_i(-q)^{C_{ii}}(a^2;q^2)^{-1}_1$.
We can write the resulting series as
\begin{equation}
\begin{aligned}
\frac{1}{(a^2;q^2)_{\infty}}P_{Q^+_1}(\boldsymbol{x},q) = &\ 1 + \frac{
x_1+\frac{1}{1-a^2}\left(
\frac{x_2}{q^2}-\frac{x_3}{q}-qx_4+q^2x_5
\right)
}{(q^2;q^2)_1} + O(x^2) \\ 
= &\ 1 + \frac{1-a^2+a^{-2}-q^{-2}-q^2+a^2}{(q^2;q^2)_1(a^2;q^2)_1}x + O(x^2).
\end{aligned}
\end{equation}
After some cancellations, we get, as expected, that the numerator of the linear term in $x$ is a~finite polynomial for Hopf link:
\begin{equation}
    P^{\text{fin}}(\mathrm{L2a1}) = \frac{1}{a^2}-q^2-\frac{1}{q^2}+1\,.
\end{equation}

\subsubsection{\texorpdfstring{$p>1: 5_{2}$, $7_{2}$, $9_{2}$}{52, 72, 92}, etc.} 


We start with $K_1=5_2$. The three-twist knot has two twist regions -- one formed by a~single bigon with parallel orientation of stands, and the other by two connected bigons with anti-parallel orientation of strands. Twist knot family is generated by adding full twists in the gray region (Figure \ref{fig:5_2 with tau2}).\footnote{Of course, it is completely valid to add full twists in the other twist region of $5_2$. We expect that in this case it will amount to application of $(1,2)$-splitting, following Conjecture \ref{coj:knot quivers twists}.}

\begin{figure}[ht!]
    \centering
    \includegraphics[width=0.25\linewidth]{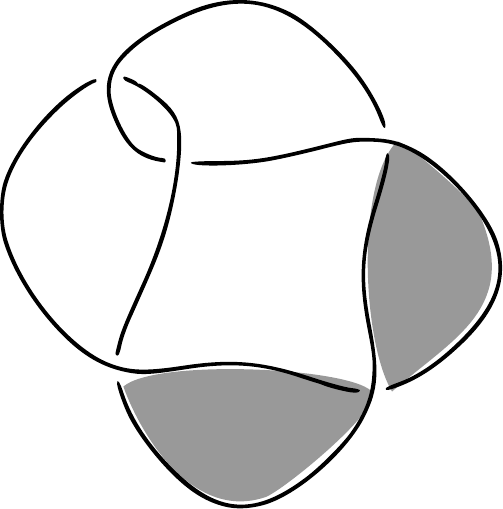}
    \caption{Knot $5_2$. The relevant twist region $\tau$ shown in gray.}
    \label{fig:5_2 with tau2}
\end{figure}

\noindent Without loss of generality, as a~starting point we pick a~particular quiver $Q_1=Q_{5_2}$ from \cite{KRSS1707long}:
\begin{equation}\label{eq:5_2 quiver}
\renewcommand\arraystretch{1.5
}
Q_{5_2} = 
\left[
\begin{array}{ccccccc}
  3 & 1 & 2 &  1 & 3 & 2 & 3 \\
  1 & 0 & 1 &  0 & 2 & 0 & 1 \\
  2 & 1 & 2 &  1 & 2 & 1 & 2 \\
   1 & 0 & 1 &  1 & 2 & 1 & 2 \\
  3 & 2 & 2 &  2 & 4 & 3 & 4 \\
  2 & 0 & 1 &  1 & 3 & 2 & 3 \\
  3 & 1 & 2 & 2 & 4 & 3 & 5 \\
\end{array}
\right],
\left[
\begin{array}{c}
 {a^4 q^{-3} x} \\
 {a^2 q^{-2} x} \\
 a^2 x \\
 {a^2 q^{-1} x} \\
 {a^4 q^{-2} x} \\
 {a^4 q^{-4} x} \\
 {a^6 q^{-5} x} \\
\end{array}
\right].
\end{equation}

\noindent
In order to find the correct augmentation $Q^+_{1}$, note that removing a~full twist from $K_1=5_2$ gives $K_0=3_1$. Using the block decomposition (\ref{eq:K_1 block form with K_0}), we confirm that $Q^{f=0}_{3_1}$ is a~subquiver in $Q_{5_2}$.
As a~result, the augmented quiver is given by
\begin{equation}\label{eq:aug quiver 5_2, one twist}
\renewcommand\arraystretch{1.2
}
Q^+_1 = 
\left[
\begin{array}{c|ccc:cccc}
 1 & 0 & 0 & 0 & 1 & 1 & 1 & 1 \\
 \hline
 0 & 3 & 1 & 2 & 1 & 3 & 2 & 3 \\
 0 & 1 & 0 & 1 & 0 & 2 & 0 & 1 \\
 0 & 2 & 1 & 2 & 1 & 2 & 1 & 2 \\
 \hdashline
 1 & 1 & 0 & 1 & 1 & 2 & 1 & 2 \\
 1 & 3 & 2 & 2 & 2 & 4 & 3 & 4 \\
 1 & 2 & 0 & 1 & 1 & 3 & 2 & 3 \\
 1 & 3 & 1 & 2 & 2 & 4 & 3 & 5 \\
\end{array}
\right],
\left[
\begin{array}{c}
 a^2 q^{-1} \\
 \hline
 {a^4 q^{-3} x} \\
 {a^2 q^{-2} x} \\
 a^2 x \\
 \hdashline
 {a^2 q^{-1} x} \\
 {a^4 q^{-2} x} \\
 {a^4 q^{-4} x} \\
 {a^6 q^{-5} x} \\
\end{array}
\right].
\end{equation}

\noindent
Using this, we can now test the Conjecture \ref{coj:knot quivers twists}. Define
\begin{equation}\label{eq:positive_twist_knots_unlinkings}
    Q^{+}_{i} = U(0,4i-1)U(0,4i-2)U(0,4i-3)U(0,4i-4)\,Q^+_{i-1},\quad i=2,3,\dots
\end{equation}
We can now compare $Q_i$ with $Q_{K_i}$. For example,
\begin{align}\label{eq:7_2 augmented quiver}
    Q^{+}_{2} = & \ U(0,7)U(0,6)U(0,5)U(0,4)\,Q^+_{1} \\
= &\
\renewcommand\arraystretch{1.2}
\left[
\begin{array}{c|ccc:cccc:cccc}
 1 & 0 & 0 & 0 & 0 & 0 & 0 & 0 & 1 & 1 & 1 & 1 \\
 \hline
 0 & 3 & 1 & 2 & 1 & 3 & 2 & 3 & 1 & 3 & 2 & 3 \\
 0 & 1 & 0 & 1 & 0 & 2 & 0 & 1 & 0 & 2 & 0 & 1 \\
 0 & 2 & 1 & 2 & 1 & 2 & 1 & 2 & 1 & 2 & 1 & 2 \\
  \hdashline
 0 & 1 & 0 & 1 & 1 & 2 & 1 & 2 & 1 & 2 & 1 & 2 \\
 0 & 3 & 2 & 2 & 2 & 4 & 3 & 4 & 3 & 4 & 3 & 4 \\
 0 & 2 & 0 & 1 & 1 & 3 & 2 & 3 & 2 & 4 & 2 & 3 \\
 0 & 3 & 1 & 2 & 2 & 4 & 3 & 5 & 3 & 5 & 4 & 5 \\
 \hdashline
 1 & 1 & 0 & 1 & 1 & 3 & 2 & 3 & 3 & 4 & 3 & 4 \\
 1 & 3 & 2 & 2 & 2 & 4 & 4 & 5 & 4 & 6 & 5 & 6 \\
 1 & 2 & 0 & 1 & 1 & 3 & 2 & 4 & 3 & 5 & 4 & 5 \\
 1 & 3 & 1 & 2 & 2 & 4 & 3 & 5 & 4 & 6 & 5 & 7 \\
\end{array}
\right],
\left[
\begin{array}{c}
 a^2 q^{-1} \\
 \hline
 {a^4 q^{-3} x} \\
 {a^2 q^{-2} x} \\
 a^2 x \\
 \hdashline
 {a^2 q^{-1} x} \\
 {a^4 q^{-2} x} \\
 {a^4 q^{-4} x} \\
 {a^6 q^{-5} x} \\
 \hdashline
 a^4 q^{-3} x \\
 a^6 q^{-4} x \\
 a^6 q^{-6} x \\
 a^8 q^{-7} x \\
\end{array}
\right].
\end{align}
The resulting quiver $Q_2$ agrees with the~quiver $Q_{7_2}$ corresponding to the knot $K_2=7_2$ \cite{KRSS1707long}.
Note that $Q_2^+$ has four linked units in the first row/column, which enables us to continue unlinking and to obtain quivers for any twist knot from this family.  Comparing \eqref{eq:positive_twist_knots_unlinkings} with the recursive construction of quivers corresponding to twist knots $5_{2}$, $7_{2}$, $9_{2}$, etc. in \cite{KRSS1707long}, we can see that for this infinite family Conjecture \ref{coj:knot quivers twists} is true.

In analogy to the case of twist knots $4_{1}$, $6_{1}$, $8_{1}$, etc., we can use Theorem \ref{thm:BPS_spectrum_of_augmented_quiver} to encode the BPS spectrum of all twist knots $5_{2}$, $7_{2}$, $9_{2}$, etc. in a~finite augmented quiver. The count of BPS states in $\Theta_{Q^+_{5_{2}}}=\Theta_{Q^+_{7_{2}}}=\dots=\Theta_{Q^+_{\infty}}$ is given by
\begin{equation}
    N_{5_2^+}(x,a,q) = a^2 + \left(a^2 \left(1 - \frac{1}{q^2} - q^2\right)+\mathcal{O}(a^4)\right) x+\mathcal{O}(x^2)\,.\label{LMOV52plus}
\end{equation}
Similarly to the previous case, we can obtain $N_{K_i}(x,a,q)$ by setting $x_0=0$ instead of $x_0=a^2 q^{-1}$ in $Q^+_{K_{i}}$ obtained from $Q^+_{5_{2}}$ by performing appropriate unlinkings. Finally, in the infinite limit we have $\Theta_{Q_{\infty}}$ which is equal to $\Theta_{Q^+_{5_2}}$ minus a~single state coming from the extra node, so the corresponding count of BPS states is given by 
\begin{equation}
    N_{\infty}(x,a,q) = \left(a^2 \left(1 - \frac{1}{q^2} - q^2\right)+\mathcal{O}(a^4)\right) x+\mathcal{O}(x^2)\,.
\end{equation}

\subsection{Torus knots}

\subsubsection{\texorpdfstring{$3_{1}$, $5_{1}$, $7_{1}$}{31, 51, 71}, etc.}

The knots in this family are $(2,2p+1)$ torus knots: $p=1$ corresponds to trefoil $3_1$, $p=2$ to cinquefoil~$5_1$,~etc.
We start with $K_1=3_1$ (here we consider the left-handed trefoil). The trefoil knot has only one twist region formed by three connected bigons with parallel orientation of stands, which suggests the use of (\ref{eq:link augmentation}). Torus knots are generated by adding any number of full twists in this region (Figure \ref{fig:3_1 with tau1}).

\begin{figure}[ht!]
    \centering
    \includegraphics[width=0.2\linewidth]{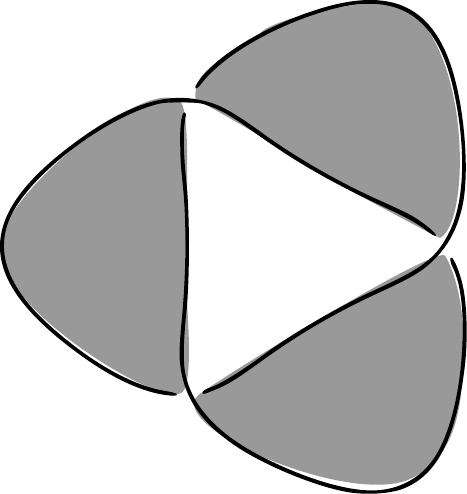}
    \caption{Knot $3_1$. The relevant twist region $\tau$ shown in gray.}
    \label{fig:3_1 with tau1}
\end{figure}

The quiver $Q_1=Q_{3_1}$ is given by \cite{KRSS1707long,JKLNS2212}

\begin{equation}\renewcommand\arraystretch{1.2}
Q_{3_1} = 
\left[
\begin{array}{ccc}
 0 & 1 & 1 \\
 1 & 2 & 2 \\
 1 & 2 & 3 \\
\end{array}
\right],
\left[
\begin{array}{c}
 {a^2 q^{-2} x} \\
 a^2 x\\
 {a^4 q^{-3} x} \\
\end{array}
\right].
\end{equation}

\noindent In this case $K_0=\bigcirc$, and the augmented quiver is chosen as (with the subquiver corresponding to the unknot separated by dashed lines):
\begin{equation}\renewcommand\arraystretch{1.2}
Q^+_{1} = 
\left[
\begin{array}{c|c:cc}
 0 & 2 & 1 & 1 \\
 \hline
 2 & 0 & 1 & 1 \\
 \hdashline
 1 & 1 & 2 & 2 \\
 1 & 1 & 2 & 3 \\
\end{array}
\right],
\left[
\begin{array}{c}
 q^2 \\
 \hline
 {a^2 q^{-2} x} \\
 \hdashline
 a^2 x\\
 {a^4 q^{-3} x} \\
\end{array}
\right].
\end{equation}
Following Conjecture \ref{coj:knot quivers twists}, we apply linking: $Q_2^+ = L(03)L(02)Q^+_{1}$. The resulting quiver is given by
\begin{equation}\renewcommand\arraystretch{1.2}
Q^+_{2} =
\left[
\begin{array}{c|c:cc:cc}
 0 & 2 & 2 & 2 & 1 & 1 \\
 \hline
 2 & 0 & 1 & 1 & 3 & 3 \\
 \hdashline
 2 & 1 & 2 & 2 & 3 & 4 \\
 2 & 1 & 2 & 3 & 3 & 4 \\
 \hdashline
 1 & 3 & 3 & 3 & 4 & 4 \\
 1 & 3 & 4 & 4 & 4 & 5 \\
\end{array}
\right],
\left[
\begin{array}{c}
 q^2 \\
 \hline
 {a^2 q^{-2} x} \\
 \hdashline
 a^2 x\\
 {a^4 q^{-3} x} \\
 \hdashline
 a^2 q^2 x\\
 {a^4 q^{-1} x} \\
\end{array}
\right]\times a^2q^{-2}.
\end{equation}
Note that by removing the extra node, we find $Q_2 = Q_{5_1}$. The next iteration is given by $Q_3^+ = L(05)L(04)Q^+_{2}$:
\begin{equation}\renewcommand\arraystretch{1.2}
Q_3^+ = 
\left[
\begin{array}{c|c:cc:cc:cc}
 0 & 2 & 2 & 2 & 2 & 2 & 1 & 1 \\
 \hline
 2 & 0 & 1 & 1 & 3 & 3 & 5 & 5 \\
 \hdashline
 2 & 1 & 2 & 2 & 3 & 4 & 5 & 6 \\
 2 & 1 & 2 & 3 & 3 & 4 & 5 & 6 \\
 \hdashline
 2 & 3 & 3 & 3 & 4 & 4 & 5 & 6 \\
 2 & 3 & 4 & 4 & 4 & 5 & 5 & 6 \\
 \hdashline
 1 & 5 & 5 & 5 & 5 & 5 & 6 & 6 \\
 1 & 5 & 6 & 6 & 6 & 6 & 6 & 7 \\
\end{array}
\right],
\left[
\begin{array}{c}
 q^2 \\
 \hline
 {a^2 q^{-2} x} \\
 \hdashline
 a^2 x\\
 {a^4 q^{-3} x} \\
 \hdashline
 a^2 q^2 x\\
 {a^4 q^{-1} x} \\
 \hdashline
 a^4 q^2 x\\
 a^6 q^{-1} x
\end{array}
\right]\times a^4q^{-4},
\end{equation}
and we can see that $Q_3 = Q_{7_1}$.
Iterating this procedure, we confirm that
\begin{equation}
    Q^{+}_{T(2,2p+3)} = L(0,2p+1)L(0,2p)\,Q^+_{T(2,2p+1)}.
\end{equation}

\subsection{Pretzel knots with odd number of twists}\label{sec:Pretzel_knots}

Pretzel knots are a~class of knots constructed from sequences of integer parameters, each encoding a~local twisting. Their diagrammatic simplicity belies a~rich structure, making them central examples in the study of knot invariants and topological classification.

\begin{defn}
    Let \( n \in \mathbb{N} \) and let \( p_1, p_2, \dots, p_n \in \mathbb{Z} \). The \emph{\(n\)-pretzel knot} associated with the tuple \( (p_1, p_2, \dots, p_n) \) is denoted by \( L(p_1, p_2, \dots, p_n) \) and is defined as the knot in \( S^3 \) obtained by joining \( n \) vertical tangles in sequence, where the \( i \)-th tangle contributes \( p_i \) half-twists (positive for right-handed, negative for left-handed crossings). 
\end{defn}

\begin{figure}[ht]
    \centering
    \includegraphics[width=0.6\linewidth]{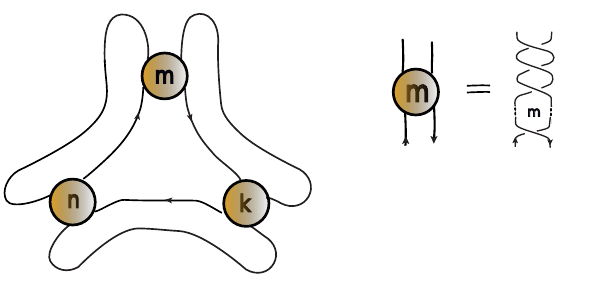}
    \caption{A classical pretzel knot}
    \label{pr3}
\end{figure}

When $n = 3$, the knot $L(m, n, k)$ is called a~\emph{classical pretzel knot}, as shown in Figure~\ref{pr3}. Notable subfamilies of pretzel knots include the \emph{twist knots} $ L(2p_1-1, 1, 1) \equiv K_{p_1>0}$ and the \emph{double twist knots} \( L(2p_1+1, 2p_2+1, 1) \).

\subsubsection{Double twist knots \texorpdfstring{$L(2k+1,3,1)$}{L(2k+1,3,1)}}

This family includes knots $7_4$, $9_5$, and so on.
We start with $K_1=7_4$ which corresponds to $k=1$.
Its quiver can be computed, for example, via the tangle addition algorithm \cite{SW1711}. Since $K_0=5_2$, we can use the knowledge of $Q_{5_2}$ to find block form (\ref{eq:K_1 block form with K_0}), and determine the corresponding augmented quiver. The result is given by
\begin{equation}\label{eq:7_4 augmented quiver}
\renewcommand\arraystretch{1.2}
Q^{+}_{7_4} = 
\left[
\begin{array}{c|ccccccc:cccccccc}
 1 & 0 & 0 & 0 & 0 & 0 & 0 & 0 & 1 & 1 & 1 & 1 & 1 & 1 & 1 & 1 \\
 \hline
 0 & 3 & 1 & 2 & 1 & 3 & 2 & 3 & 1 & 3 & 2 & 3 & 1 & 3 & 2 & 3 \\
 0 & 1 & 0 & 1 & 0 & 2 & 0 & 1 & 0 & 2 & 0 & 1 & 0 & 2 & 0 & 1 \\
 0 & 2 & 1 & 2 & 1 & 2 & 1 & 2 & 1 & 2 & 1 & 2 & 1 & 2 & 1 & 2 \\
 0 & 1 & 0 & 1 & 1 & 2 & 1 & 2 & 0 & 2 & 1 & 2 & 1 & 2 & 1 & 2 \\
 0 & 3 & 2 & 2 & 2 & 4 & 3 & 4 & 1 & 3 & 3 & 4 & 3 & 4 & 3 & 4 \\
 0 & 2 & 0 & 1 & 1 & 3 & 2 & 3 & 0 & 2 & 1 & 3 & 2 & 4 & 2 & 3 \\
 0 & 3 & 1 & 2 & 2 & 4 & 3 & 5 & 1 & 3 & 2 & 4 & 3 & 5 & 4 & 5 \\
 \hdashline
 1 & 1 & 0 & 1 & 0 & 1 & 0 & 1 & 1 & 2 & 1 & 2 & 1 & 2 & 1 & 2 \\
 1 & 3 & 2 & 2 & 2 & 3 & 2 & 3 & 2 & 4 & 3 & 4 & 3 & 4 & 3 & 4 \\
 1 & 2 & 0 & 1 & 1 & 3 & 1 & 2 & 1 & 3 & 2 & 3 & 2 & 4 & 2 & 3 \\
 1 & 3 & 1 & 2 & 2 & 4 & 3 & 4 & 2 & 4 & 3 & 5 & 3 & 5 & 4 & 5 \\
 1 & 1 & 0 & 1 & 1 & 3 & 2 & 3 & 1 & 3 & 2 & 3 & 3 & 4 & 3 & 4 \\
 1 & 3 & 2 & 2 & 2 & 4 & 4 & 5 & 2 & 4 & 4 & 5 & 4 & 6 & 5 & 6 \\
 1 & 2 & 0 & 1 & 1 & 3 & 2 & 4 & 1 & 3 & 2 & 4 & 3 & 5 & 4 & 5 \\
 1 & 3 & 1 & 2 & 2 & 4 & 3 & 5 & 2 & 4 & 3 & 5 & 4 & 6 & 5 & 7 \\
\end{array}
\right],~\left[
\begin{array}{c}
a^2 q^{-1} \\
\hline
a^4 q^{-3} x\\
a^2 q^{-2} x\\
a^2 x\\ 
a^2 q^{-1} x\\
a^4 q^{-2} x\\
a^4 q^{-4} x\\
a^6 q^{-5} x\\
 \hdashline
a^2 q^{-1} x\\
a^4 q^{-2} x\\
a^4 q^{-4} x\\
a^6 q^{-5} x\\
a^4 q^{-3} x\\
a^6 q^{-4} x\\
a^6 q^{-6} x\\
a^8 q^{-7} x
\end{array}
\right]
.\end{equation}
(Here $Q^{f=0}_{5_2}$ is the 7 by 7 subquiver separated by dashed lines.)
Applying unlinking to the augmented quiver above, we confirm that quivers for double twist knots are generated from
\begin{equation}
    Q^{+}_{K_{i+1}} = U(1,16i)U(1,16i-1)\cdots U(1,16i-7)\,Q^+_{K_i}.
\end{equation}
This is perfectly consistent with the previous results \cite{Singh:2023qpw}. For example, computing $U(1,16)\dots U(1,9)\,Q^+_{7_4}$ gives
\begin{align}\label{eq:9_5 augmented quiver}
    Q^{+}_{9_5} =  
\renewcommand\arraystretch{1.2}
\Resize{12.5cm}{
\left[
\begin{array}{c|ccccccc:cccccccc:cccccccc}
 1 & 0 & 0 & 0 & 0 & 0 & 0 & 0 & 0 & 0 & 0 & 0 & 0 & 0 & 0 & 0 & 1 & 1 & 1 & 1 & 1 & 1 & 1 & 1 \\
 \hline
 0 & 3 & 1 & 2 & 1 & 3 & 2 & 3 & 1 & 3 & 2 & 3 & 1 & 3 & 2 & 3 & 1 & 3 & 2 & 3 & 1 & 3 & 2 & 3 \\
 0 & 1 & 0 & 1 & 0 & 2 & 0 & 1 & 0 & 2 & 0 & 1 & 0 & 2 & 0 & 1 & 0 & 2 & 0 & 1 & 0 & 2 & 0 & 1 \\
 0 & 2 & 1 & 2 & 1 & 2 & 1 & 2 & 1 & 2 & 1 & 2 & 1 & 2 & 1 & 2 & 1 & 2 & 1 & 2 & 1 & 2 & 1 & 2 \\
 0 & 1 & 0 & 1 & 1 & 2 & 1 & 2 & 0 & 2 & 1 & 2 & 1 & 2 & 1 & 2 & 0 & 2 & 1 & 2 & 1 & 2 & 1 & 2 \\
 0 & 3 & 2 & 2 & 2 & 4 & 3 & 4 & 1 & 3 & 3 & 4 & 3 & 4 & 3 & 4 & 1 & 3 & 3 & 4 & 3 & 4 & 3 & 4 \\
 0 & 2 & 0 & 1 & 1 & 3 & 2 & 3 & 0 & 2 & 1 & 3 & 2 & 4 & 2 & 3 & 0 & 2 & 1 & 3 & 2 & 4 & 2 & 3 \\
 0 & 3 & 1 & 2 & 2 & 4 & 3 & 5 & 1 & 3 & 2 & 4 & 3 & 5 & 4 & 5 & 1 & 3 & 2 & 4 & 3 & 5 & 4 & 5 \\
 \hdashline
 0 & 1 & 0 & 1 & 0 & 1 & 0 & 1 & 1 & 2 & 1 & 2 & 1 & 2 & 1 & 2 & 1 & 2 & 1 & 2 & 1 & 2 & 1 & 2 \\
 0 & 3 & 2 & 2 & 2 & 3 & 2 & 3 & 2 & 4 & 3 & 4 & 3 & 4 & 3 & 4 & 3 & 4 & 3 & 4 & 3 & 4 & 3 & 4 \\
 0 & 2 & 0 & 1 & 1 & 3 & 1 & 2 & 1 & 3 & 2 & 3 & 2 & 4 & 2 & 3 & 2 & 4 & 2 & 3 & 2 & 4 & 2 & 3 \\
 0 & 3 & 1 & 2 & 2 & 4 & 3 & 4 & 2 & 4 & 3 & 5 & 3 & 5 & 4 & 5 & 3 & 5 & 4 & 5 & 3 & 5 & 4 & 5 \\
 0 & 1 & 0 & 1 & 1 & 3 & 2 & 3 & 1 & 3 & 2 & 3 & 3 & 4 & 3 & 4 & 2 & 4 & 3 & 4 & 3 & 4 & 3 & 4 \\
 0 & 3 & 2 & 2 & 2 & 4 & 4 & 5 & 2 & 4 & 4 & 5 & 4 & 6 & 5 & 6 & 3 & 5 & 5 & 6 & 5 & 6 & 5 & 6 \\
 0 & 2 & 0 & 1 & 1 & 3 & 2 & 4 & 1 & 3 & 2 & 4 & 3 & 5 & 4 & 5 & 2 & 4 & 3 & 5 & 4 & 6 & 4 & 5 \\
 0 & 3 & 1 & 2 & 2 & 4 & 3 & 5 & 2 & 4 & 3 & 5 & 4 & 6 & 5 & 7 & 3 & 5 & 4 & 6 & 5 & 7 & 6 & 7 \\
 \hdashline
 1 & 1 & 0 & 1 & 0 & 1 & 0 & 1 & 1 & 3 & 2 & 3 & 2 & 3 & 2 & 3 & 3 & 4 & 3 & 4 & 3 & 4 & 3 & 4 \\
 1 & 3 & 2 & 2 & 2 & 3 & 2 & 3 & 2 & 4 & 4 & 5 & 4 & 5 & 4 & 5 & 4 & 6 & 5 & 6 & 5 & 6 & 5 & 6 \\
 1 & 2 & 0 & 1 & 1 & 3 & 1 & 2 & 1 & 3 & 2 & 4 & 3 & 5 & 3 & 4 & 3 & 5 & 4 & 5 & 4 & 6 & 4 & 5 \\
 1 & 3 & 1 & 2 & 2 & 4 & 3 & 4 & 2 & 4 & 3 & 5 & 4 & 6 & 5 & 6 & 4 & 6 & 5 & 7 & 5 & 7 & 6 & 7 \\
 1 & 1 & 0 & 1 & 1 & 3 & 2 & 3 & 1 & 3 & 2 & 3 & 3 & 5 & 4 & 5 & 3 & 5 & 4 & 5 & 5 & 6 & 5 & 6 \\
 1 & 3 & 2 & 2 & 2 & 4 & 4 & 5 & 2 & 4 & 4 & 5 & 4 & 6 & 6 & 7 & 4 & 6 & 6 & 7 & 6 & 8 & 7 & 8 \\
 1 & 2 & 0 & 1 & 1 & 3 & 2 & 4 & 1 & 3 & 2 & 4 & 3 & 5 & 4 & 6 & 3 & 5 & 4 & 6 & 5 & 7 & 6 & 7 \\
 1 & 3 & 1 & 2 & 2 & 4 & 3 & 5 & 2 & 4 & 3 & 5 & 4 & 6 & 5 & 7 & 4 & 6 & 5 & 7 & 6 & 8 & 7 & 9 \\
\end{array}
\right]~~~\left[
\begin{array}{c}
a^2 q^{-1} \\ \hline
a^4 q^{-3} x\\
a^2 q^{-2} x\\
a^2 x\\
a^2 q^{-1} x\\
a^4 q^{-2} x\\
a^4 q^{-4} x\\
a^6 q^{-5} x\\ \hdashline
a^2 q^{-1} x\\
a^4 q^{-2} x\\
a^4 q^{-4} x\\
a^6 q^{-5} x\\
a^4 q^{-3} x\\
a^6 q^{-4} x\\
a^6 q^{-6} x\\
a^8 q^{-7} x\\ \hdashline
a^4 q^{-3} x\\
a^6 q^{-4} x\\
a^6 q^{-6} x\\
a^8 q^{-7} x\\
a^6 q^{-5} x\\
a^8 q^{-6} x\\
a^8 q^{-8} x\\
a^{10} q^{-9} x
\end{array}
\right].
}
\nonumber
\nonumber
\end{align}
Compared to the Melvin-Morton-Rozansky (MMR) expansion method \cite{Singh:2023qpw}, the current approach is far more streamlined, avoiding the twofold complexity of classical inverse binomial expansion followed by quantum deformation mapping. Additionally, the MMR framework suffers from the difficulty of determining too many unknown parameters from a~finite set of equations.

Moreover, Theorem \ref{thm:BPS_spectrum_of_augmented_quiver} implies that the LMOV spectrum for the augmented quiver $Q^+_{7_{4}}$ generates the spectrum for any knot within this family, analogously to the earlier examples.

\subsubsection{Pretzel knots with wider homology and vertical growth}

Another interesting example is the family of pretzel knots $K_1=8_{13}$, $K_2=10_{10}$, $K_3=12a_{744}$, and so~on. The corresponding twist region has the anti-parallel orientation of strands, and we anticipate the application of $(0,2)$-splitting.
We also note that undoing a~full twist in $8_{13}$ produces $K_0=6_3$.

First, we find $Q_{8_{13}}$ by using the tangle addition algorithm \cite{SW1711}. Second, we augment this quiver basing on the block decomposition (\ref{eq:K_1 block form with K_0}) and the knowledge of $Q^{f=0}_{6_3}$, which we combine with (\ref{eq:unlink augmentation}). The resulting quiver is given by
\begin{equation*}
Q_{8_{13}}^+ = 
\renewcommand\arraystretch{1.2}
\Resize{12.5cm}{
\setlength\arraycolsep{1pt}
\left[
\begin{array}{c|ccccccccccccc:cccccccccccccccc}
 1 & 0 & 0 & 0 & 0 & 0 & 0 & 0 & 0 & 0 & 0 & 0 & 0 & 0 & 1 & 1 & 1 & 1 & 1 & 1 & 1 & 1 & 1 & 1 & 1 & 1 & 1 & 1 & 1 & 1 \\
 \hline
 0 & 0 & 0 & 0 & -1 & -1 & 1 & 0 & 0 & -1 & 0 & 0 & -1 & -1 & 0 & 0 & -1 & -1 & 1 & 1 & 0 & 0 & 0 & 0 & 0 & -1 & -1 & 0 & -1 & -1 \\
 0 & 0 & 1 & 0 & -1 & -2 & 1 & 0 & -1 & -1 & 1 & 1 & 0 & -1 & 1 & 1 & 0 & -1 & 2 & 1 & 1 & -1 & 1 & 2 & 1 & 0 & 0 & 0 & 0 & -2 \\
 0 & 0 & 0 & 0 & -1 & -2 & 0 & 0 & -1 & -2 & 1 & 1 & 0 & -1 & 0 & 1 & -1 & -1 & 0 & 1 & -1 & -1 & 1 & 2 & 0 & 0 & 0 & 0 & -1 & -2 \\
 0 & -1 & -1 & -1 & -2 & -3 & 0 & -1 & -2 & -2 & -1 & 0 & -2 & -2 & -1 & 0 & -2 & -2 & 0 & 1 & -1 & -2 & -1 & 0 & -1 & -1 & -2 & 0 & -2 & -3 \\
 0 & -1 & -2 & -2 & -3 & -3 & -2 & -1 & -3 & -3 & -1 & 0 & -2 & -2 & -2 & -1 & -3 & -2 & -2 & -1 & -3 & -2 & -1 & 0 & -2 & -1 & -2 & -1 & -3 & -3 \\
 0 & 1 & 1 & 0 & 0 & -2 & 2 & 1 & 0 & -1 & 2 & 2 & 1 & -1 & 1 & 2 & 0 & 0 & 2 & 2 & 1 & 0 & 2 & 3 & 1 & 1 & 1 & 0 & 0 & -2 \\
 0 & 0 & 0 & 0 & -1 & -1 & 1 & 1 & 0 & -1 & 1 & 1 & 0 & -1 & 0 & 0 & -1 & -1 & 1 & 1 & 0 & 0 & 1 & 1 & 0 & 0 & 0 & 0 & -1 & -1 \\
 0 & 0 & -1 & -1 & -2 & -3 & 0 & 0 & -1 & -2 & 0 & 1 & -1 & -2 & -1 & 0 & -2 & -1 & 0 & 1 & -1 & -1 & 0 & 1 & -1 & 0 & -1 & 0 & -2 & -3 \\
 0 & -1 & -1 & -2 & -2 & -3 & -1 & -1 & -2 & -2 & 0 & 0 & -1 & -2 & -1 & -1 & -2 & -2 & -1 & -1 & -2 & -2 & 0 & 0 & -2 & -1 & -1 & -2 & -3 & -3 \\
 0 & 0 & 1 & 1 & -1 & -1 & 2 & 1 & 0 & 0 & 3 & 2 & 1 & 0 & 2 & 1 & 1 & -1 & 3 & 2 & 2 & 0 & 3 & 3 & 2 & 1 & 2 & 1 & 1 & -1 \\
 0 & 0 & 1 & 1 & 0 & 0 & 2 & 1 & 1 & 0 & 2 & 2 & 1 & 0 & 1 & 1 & 0 & 0 & 2 & 2 & 1 & 1 & 2 & 2 & 1 & 1 & 1 & 1 & 0 & 0 \\
 0 & -1 & 0 & 0 & -2 & -2 & 1 & 0 & -1 & -1 & 1 & 1 & 0 & -1 & 0 & 1 & -1 & -2 & 1 & 2 & 0 & -1 & 1 & 2 & 0 & 0 & 0 & 1 & -1 & -2 \\
 0 & -1 & -1 & -1 & -2 & -2 & -1 & -1 & -2 & -2 & 0 & 0 & -1 & -1 & -1 & -1 & -2 & -2 & -1 & -1 & -2 & -2 & 0 & 0 & -1 & -1 & -1 & -1 & -2 & -2 \\
\hdashline
 1 & 0 & 1 & 0 & -1 & -2 & 1 & 0 & -1 & -1 & 2 & 1 & 0 & -1 & 3 & 2 & 1 & 0 & 3 & 2 & 2 & 0 & 3 & 3 & 2 & 1 & 1 & 1 & 1 & -1 \\
 1 & 0 & 1 & 1 & 0 & -1 & 2 & 0 & 0 & -1 & 1 & 1 & 1 & -1 & 2 & 2 & 1 & 0 & 3 & 2 & 2 & 0 & 2 & 2 & 2 & 1 & 1 & 1 & 1 & -1 \\
 1 & -1 & 0 & -1 & -2 & -3 & 0 & -1 & -2 & -2 & 1 & 0 & -1 & -2 & 1 & 1 & 0 & -1 & 1 & 1 & 0 & -1 & 2 & 2 & 0 & 0 & 0 & 0 & -1 & -2 \\
 1 & -1 & -1 & -1 & -2 & -2 & 0 & -1 & -1 & -2 & -1 & 0 & -2 & -2 & 0 & 0 & -1 & -1 & 1 & 1 & 0 & -1 & 0 & 0 & 0 & -1 & -1 & 0 & -1 & -2 \\
 1 & 1 & 2 & 0 & 0 & -2 & 2 & 1 & 0 & -1 & 3 & 2 & 1 & -1 & 3 & 3 & 1 & 1 & 4 & 3 & 2 & 1 & 4 & 4 & 2 & 2 & 2 & 1 & 1 & -1 \\
 1 & 1 & 1 & 1 & 1 & -1 & 2 & 1 & 1 & -1 & 2 & 2 & 2 & -1 & 2 & 2 & 1 & 1 & 3 & 3 & 2 & 1 & 3 & 3 & 2 & 2 & 2 & 1 & 1 & -1 \\
 1 & 0 & 1 & -1 & -1 & -3 & 1 & 0 & -1 & -2 & 2 & 1 & 0 & -2 & 2 & 2 & 0 & 0 & 2 & 2 & 1 & 0 & 3 & 3 & 0 & 1 & 1 & 0 & -1 & -2 \\
 1 & 0 & -1 & -1 & -2 & -2 & 0 & 0 & -1 & -2 & 0 & 1 & -1 & -2 & 0 & 0 & -1 & -1 & 1 & 1 & 0 & 0 & 1 & 1 & 0 & 0 & 0 & 0 & -1 & -2 \\
 1 & 0 & 1 & 1 & -1 & -1 & 2 & 1 & 0 & 0 & 3 & 2 & 1 & 0 & 3 & 2 & 2 & 0 & 4 & 3 & 3 & 1 & 5 & 4 & 3 & 2 & 3 & 2 & 2 & 0 \\
 1 & 0 & 2 & 2 & 0 & 0 & 3 & 1 & 1 & 0 & 3 & 2 & 2 & 0 & 3 & 2 & 2 & 0 & 4 & 3 & 3 & 1 & 4 & 4 & 3 & 2 & 3 & 2 & 2 & 0 \\
 1 & 0 & 1 & 0 & -1 & -2 & 1 & 0 & -1 & -2 & 2 & 1 & 0 & -1 & 2 & 2 & 0 & 0 & 2 & 2 & 0 & 0 & 3 & 3 & 2 & 1 & 1 & 1 & 0 & -1 \\
 1 & -1 & 0 & 0 & -1 & -1 & 1 & 0 & 0 & -1 & 1 & 1 & 0 & -1 & 1 & 1 & 0 & -1 & 2 & 2 & 1 & 0 & 2 & 2 & 1 & 1 & 1 & 1 & 0 & -1 \\
 1 & -1 & 0 & 0 & -2 & -2 & 1 & 0 & -1 & -1 & 2 & 1 & 0 & -1 & 1 & 1 & 0 & -1 & 2 & 2 & 1 & 0 & 3 & 3 & 1 & 1 & 2 & 1 & 0 & -1 \\
 1 & 0 & 0 & 0 & 0 & -1 & 0 & 0 & 0 & -2 & 1 & 1 & 1 & -1 & 1 & 1 & 0 & 0 & 1 & 1 & 0 & 0 & 2 & 2 & 1 & 1 & 1 & 1 & 0 & -1 \\
 1 & -1 & 0 & -1 & -2 & -3 & 0 & -1 & -2 & -3 & 1 & 0 & -1 & -2 & 1 & 1 & -1 & -1 & 1 & 1 & -1 & -1 & 2 & 2 & 0 & 0 & 0 & 0 & -1 & -2 \\
 1 & -1 & -2 & -2 & -3 & -3 & -2 & -1 & -3 & -3 & -1 & 0 & -2 & -2 & -1 & -1 & -2 & -2 & -1 & -1 & -2 & -2 & 0 & 0 & -1 & -1 & -1 & -1 & -2 & -2 \\
\end{array}
\right]
}
\end{equation*}
with vector of variables
\begin{equation}
\begin{aligned}
\boldsymbol{x} = &\ \left[\left. {a^2}{q^{-1}} \, \right| \,
 x,\frac{a^2 x}{q^3},x,\frac{x}{q^2},\frac{q x}{a^2},\frac{a^2 x}{q^2},q x,\frac{x}{q},\frac{q^2 x}{a^2},\frac{a^2 x}{q},q^2 x,x,\frac{q^3 x}{a^2},\frac{a^4
   x}{q^5},\frac{a^2 x}{q^2},\frac{a^2 x}{q^4},\frac{x}{q},\frac{a^4 x}{q^4},\frac{a^2 x}{q}, \right. \\
   & \qquad \qquad
   \left. \frac{a^2 x}{q^3}, x,\frac{a^4 x}{q^3},a^2 x,\frac{a^2 x}{q^2},q x,\frac{a^2
   x}{q^2},q x,\frac{x}{q},\frac{q^2 x}{a^2}
 \right].
\end{aligned}
\end{equation}
As a~result, we confirm that the quivers corresponding to knots $10_{10},12a_{744},\dots$ are generated from unlinking $Q_{8_{13}}^+$ as follows:
\begin{equation}
    Q^+_{10_{10}} = U(1,30)U(1,29)\dots U(1,15)Q^+_{8_{13}}, \dots
\end{equation}
Moreover, Theorem \ref{thm:BPS_spectrum_of_augmented_quiver} implies that the LMOV spectrum for the augmented quiver $Q^+_{8_{13}}$ generates the spectrum for any knot within this family, analogously to the earlier examples.

\subsection{Knot complements}\label{sec:knot_complements}

Another interesting scenario when quiver stability translates into something topological is the case of knot complement invariants.
As briefly summarised in Section \ref{knot-complement-quiver-section}, the invariant $F_K(x,a,q)$ can be written in terms of quiver generating series.

We will now discuss the stable behaviour of such quivers in a~special case of $(2,2p+1)$ torus knots. It is worth mentioning that Conjecture \ref{coj:knot quivers twists} does not apply here, at least not directly --
the knot complement quivers cannot be related via the skein relation, as was the case for the HOMFLY-PT quivers. Therefore, in order to describe their stable growth, we will focus on the quiver side and check whether the knot complement side follows consistently.

\subsubsection{Quivers for complements of \texorpdfstring{$3_1$, $5_1$, $7_1$}{31, 51, 71}, etc.}

We will show how to obtain the quivers corresponding to knot complements of this family using linking.

Our starting point is the following quiver for the $3_1$ knot complement \cite{Kuch2005}:
\begin{equation}\renewcommand\arraystretch{1.2}
Q_{S^3\setminus 3_1} = 
\left[
\begin{array}{cccc}
	0&-1&0&-1\\
	-1&-1&0&-1\\
	0&0&1&0\\
	-1&-1&0&0\\
\end{array}
\right],
\left[
\begin{array}{c}
	ax^2 \\
	q^3 x^2 \\
	aq^{-1} x \\
	q^2x \\
\end{array}
\right].
\end{equation}
We augment $Q_{S^3\setminus 3_1} $ with an extra node as follows:
\begin{equation}\renewcommand\arraystretch{1.2}
	Q_{S^3\setminus 3_1}^+ = 
	\left[
	\begin{array}{c|cccc}
		0&-1&-1&-1&-1\\
		\hline
		-1&0&-1&0&-1\\
		-1&-1&-1&0&-1\\
		-1&0&0&1&0\\
		-1&-1&-1&0&0\\
	\end{array}
	\right],
	\left[
	\begin{array}{c}
		x^2q^2\\
		\hline
		ax^2 \\
		q^3 x^2 \\
		aq^{-1} x \\
		q^2x \\
	\end{array}
	\right].
\end{equation}
Note that this choice is compatible with $(-1,-2)$-splitting (see Proposition \ref{prp:recursive_linking}). Performing splitting amounts in the following linking operation:
\begin{equation}\renewcommand\arraystretch{1.2}
\begin{aligned}
	&\ L(0,4)L(0,3)L(0,2)L(0,1)Q^+_{S^3\setminus 3_1} = \\
    &
    \hspace{1.5cm}
    \left[
	\begin{array}{c|cccc:cccc}
		0&0 & 0 & 0 & 0 & -1 &-1 & -1 & -1 \\
		\hline
		0&0 & -1 & 0 & -1 & -1 & -1 & 0 & -1 \\
		0&-1 & -1 & 0 & -1 & -2 & -2 & 0 & -1 \\
		0&0 & 0 & 1 & 0 & -1 & -1 & 0 & 0 \\
		0&-1 & -1 & 0 & 0 & -2 & -2 & -1 & -1 \\
		\hdashline
		-1&-1 & -2 & -1 & -2 & -2 & -3 & -2 & -3 \\
		-1&-1 & -2 & -1 & -2 & -3 & -3 & -2 & -3 \\
		-1&0 & 0 & 0 & -1 & -2 & -2 & -1 & -2 \\
		-1&-1 & -1 & 0 & -1 & -3 & -3 & -2 & -2 \\
	\end{array}
	\right],
	\left[
	\begin{array}{c}
		x^2 q^2\\
		\hline
		ax^2 \\
		q^3x^2 \\
		aq^{-1}x \\
		q^2x \\
		\hdashline
		aq^2x^4 \\
		q^5x^4 \\
		a q x^3 \\
		q^4 x^3 \\
	\end{array}
	\right].
\end{aligned}
\end{equation}
Comparing with \cite{Kuch2005}, we can see that after removing the augmented node the resulting quiver corresponds to quiver $Q_{S^3\setminus 5_1}$ for the $5_1$ knot complement. We can then keep this augmented node and define
\begin{equation}
    Q_{S^3\setminus 5_1}^{+} = 
	L(0,4)L(0,3)L(0,2)L(0,1)Q^+_{S^3\setminus 3_1}.
\end{equation}
Repeating the same step again, we get
\begin{equation}\renewcommand\arraystretch{1.2}
\begin{aligned}
	&\ L(0,8)L(0,7)L(0,6)L(0,5)Q^+_{S^3\setminus 5_1}\,= 
    \\
    &
    \hspace{1.5cm}
    \left[
	\begin{array}{c|cccc:cccc:cccc}
		0&0&0&0&0&0&0&0&0&-1&-1&-1&-1\\
		\hline
		0&0 & -1 & 0 & -1 & -1 &-1 & 0 & -1&-1&-1&0&-1 \\
		0&-1 & -1 & 0 & -1 & -2 &-2 & 0 & -1&-2&-2&0&-1 \\
		0&0 & 0& 1 & 0 & -1 &-1 & 0 & 0&-1&-1&0&0 \\
		0&-1 & -1 & 0 & 0 & -2 &-2 & -1 & -1&-2&-2&-1&-1 \\
		\hdashline
		0&-1 & -2 & -1 & -2 & -2 &-3 & -2 & -3&-3&-3&-2&-3 \\
		0&-1 & -2 & -1 & -2 & -3 &-3 & -2 & -3&-4&-4&-2&-3 \\
		0&0 & 0 & 0 & -1 & -2 &-2 & -1 & -2&-3&-3&-2&-2 \\
		0&-1 & -1 & 0 & -1 & -3 &-3 & -2 & -2&-4&-4&-3&-3 \\
		\hdashline
		-1&-1 & -2 & -1 & -2 & -3 &-4 &-3 & -4&-4&-5&-4&-5 \\
		-1&-1 & -2 & -1 & -2 & -3 &-4 & -3 & -4&-5&-5&-4&-5 \\
		-1&0 & 0 & 0 & -1 & -2 &-2 & -2 & -3&-4&-4&-3&-4 \\
		-1&-1 & -1 & 0 & -1 & -3 &-3 & -2 & -3&-5&-5&-4&-4 \\
	\end{array}
	\right],
	\left[
	\begin{array}{c}
		x^2 q^2\\
		\hline
		ax^2 \\
		q^3x^2 \\
		aq^{-1}x \\
		q^2x \\
		\hdashline
		aq^2x^4 \\
		q^5x^4 \\
		a q x^3 \\
		q^4 x^3 \\
		\hdashline
		aq^4x^6\\
		q^7x^6\\
		aq^3x^5\\
		q^6x^5\\
	\end{array}
	\right].
\end{aligned}
\end{equation}
Comparing with \cite{Kuch2005} we can see that after removing the augmented node the resulting quiver corresponds to quiver $Q_{S^3\setminus 7_1}$ for the $7_1$ knot complement.

In general, if we set $Q_1=Q_{S^3\setminus 3_1}$ and define $Q_{i+1}^+:=L(0,4i)L(0,4i-1)L(0,4i-2)L(0,4i-3)Q_{i}^+$ for $i=1,2,\dots$, we can see that
\begin{equation}
    Q_{i} = Q_{S^3\backslash T(2,2i+1)} \,.
\end{equation}

Summing up, we found experimentally that quivers for $(2,2p+1)$ torus knot complements grow by $(-1,-2)$-splitting. As a~result, they can be related by linking. It would be very interesting to generalise this statement to other knot complements.

\subsubsection{Analogues of LMOV invariants} \label{sec:analogues_LMOV_knot_complements}
Since quivers $Q_{S^3\backslash T(2,2p+1)}^+$ grow by $(-1,-2)$-splitting, Theorem \ref{thm:BPS_spectrum_of_augmented_quiver} guarantees that the BPS spectra of the entire family can be encoded in a~single augmented quiver $Q_{S^3\backslash 3_1}^+$.  This result is interesting on its own, since so far the analogues of LMOV invariants for knot complements were computed only for trefoil and figure-eight knot complements \cite{Kuch2005,JKLNS2212}.

The easiest way to compute the DT invariants for $Q_{S^3\backslash 3_1}^+$ -- and consequentially the analogues of LMOV invariants for the whole family $(2,2p+1)$ torus knot complements -- is to use the quiver diagonalization method, which is described in \cite{JKLNS2212}. It leads to the following DT invariants:
\begin{equation}
\begin{aligned}
\Omega(\boldsymbol{x},q) = &-x_0 - x_1 + q^{-1} x_2 + q x_3 - x_4 - \frac{x_0 x_1}{q^2} + \frac{x_0 x_2}{q^3} + 
 \frac{x_0 x_3}{q} - \frac{x_0 x_4}{q^2} + \frac{x_1 x_2}{q^3} - \frac{x_1 x_4}{q^2}\\& + \frac{x_2 x_4}{q^3} + 
 \frac{x_0^2 x_3}{q^3}  - \frac{x_0^2 x_4}{q^4} - \frac{x_0 x_1 x_4}{q^4} - \frac{x_1^2 x_4}{q^4}+ \frac{x_0 x_2 x_4}{q^5} +
  \frac{x_1 x_2 x_4}{q^5} - \frac{x_2^2 x_4}{q^8} + \frac{x_0 x_1 x_3}{q^3} \\ &- \frac{x_0 x_2 x_3}{q^4} - 
 \frac{x_0 x_1 x_4}{q^6} - \frac{x_0 x_1 x_4}{q^4} + \frac{x_0 x_2 x_4}{q^7}  + \frac{x_0 x_2 x_4}{q^5} + \frac{x_0 x_3 x_4}{q^3} -
  \frac{x_0 x_4^2}{q^4} + \frac{x_1 x_2 x_4}{q^7} \\
 & + \frac{x_1 x_2 x_4}{q^5} - \frac{x_1 x_4^2}{q^4} + 
 \frac{x_2 x_4^2}{q^5} - \frac{x_2^2 x_4}{q^6} + \frac{x_0 x_3 x_4^2}{q^5} - \frac{x_0 x_4^3}{q^6} - \frac{x_1 x_4^3}{q^6}  +
  \frac{x_2 x_4^3}{q^7}+ \dots \,,\label{DT-31C}
\end{aligned}
\end{equation}
which after considering the change of variables
\begin{equation}
    \left[
\begin{array}{c}
 x_0 \\
 \hline
 x_1 \\
 x_2 \\
 x_3 \\
 x_4 \\
\end{array}
\right]=\left[
\begin{array}{c}
 q^2 x^2 \\
 \hline
 a~x^2 \\
 q^3 x^2 \\
 aq^{-1}x \\
 q^2 x \\
\end{array}
\right],\label{knot-quiver-var-1}
\end{equation}
give analogues of LMOV invariants up to $\mathcal{O}(x^6)$:
\begin{equation}\label{eq:LMOV for 3_1 complement}
        N_{S^3\setminus 3_1^+}(x,a,q)=\left(a-q^2\right)x-a x^2+q^2x^4+\mathcal{O}(x^6).
\end{equation}
This generating series encodes analogues of LMOV invariants for the whole family $(2,2p+1)$ torus knot complements -- we can compute it to any given order and then take out appropriate states, as described at the end of Section \ref{sec:Quiver perspective}.

\paragraph{} The LMOV-type spectra for augmented knot complement quivers differs from those for augmented quivers for knots. In the knot complement case, there are only finitely many terms at each order in $x$. This happens because the change of variables for the extra node is given by $x_0=q^2x^2$, which includes a~factor of $x^2$. In contrast, for knots, the change of variables for the extra node does not depend on $x$: $x_0=aq^{-1}$. Another interesting observation is that at $a=q^2$, quiver generating function of the augmented quiver do not trivialise into 1, which is different from the case of knot complements without the extra node (see \cite[Section 7.1]{JKLNS2212} for comparison). Our analysis in this Section can be seen as the first step towards understanding transformation properties of such invariants for a~general knot complement, and we leave other questions and examples for future research.

\section{Discussion and future directions}

In this paper we examined the stable behaviour of HOMFLY-PT polynomials of knots coloured by symmetric presentations on one side, and symmetric quivers on the other. We then unified the two pictures using the knot-quiver correspondence, and conjectured the precise algorithm of constructing quivers for knots which differ by one full twist (Conjecture \ref{coj:knot quivers twists}). Our findings suggest that knot homologies of such knots are related in a~very special way. Possibly, there exists a~larger homology theory which unifies all symmetric colours into one, and whose differentials relate homologies of different knots from that given family. On the other hand, stability of quivers is expressed via operation of splitting, whose meaning for homological algebras associated to a~symmetric quiver remain mysterious. It is however expected that splitting of quivers can be expressed more explicitly using the operator form of quantum trace for rational knots, where the number of full twists can enter in some exponents inside this trace \cite{SW1711,SW-II}.
These interpretations definitely deserve closer attention, and we believe that interpreting and studying the homological side of augmented quivers for knots and their complements, splitting and unlinking will lead to many new interesting discoveries.

Another interesting observation we found is that taking the limit of infinite twists may lead to non-trivial consequences. For example, taking the direct limit of the recursion for a~pair of anti-parallel oriented strands effectively resolves the twist region, and we obtain the relation between twisted knots and a~link obtained from 0-resolution of the corresponding twist region. Whether or not (if true, then to what extent?) this holds on the homological level, we leave as an open problem. Given that there are several works on quantum invariants which focus on similar limits, we believe that considering them may also produce some interesting non-trivial results.

Summing up, our findings allowed to formulate the relation between twisting the knot and unlinking the corresponding quiver, and the next step would be to prove this relation for all rational knots, as well as to understand its homological interpretation, which we leave for future research.

\section*{Acknowledgements}
The work of S.C. is supported by the project “Quiver Structures in Knot Invariants from String Theory and Enumerative Geometry”, funded by Vergstiftelsen and the Swedish Research Council (VR 2022-06593), through the Centre of Excellence in Geometry and Physics at Uppsala University. S.C. gratefully acknowledges support from the NCN OPUS grant 2019/B/35/ST1/01120. The work of P.K. has been supported by the Polish National Science Centre through SONATA grant (2022/47/D/ST2/02058).  P.R.~would like to acknowledge the SPARC/2019-2020/ P2116/ project funding for the 4 day workshop ``Knots, Quivers and Beyond'' in February 2025 where all the authors could meet for discussion on this work.  V.K.S. is supported by the ‘Tamkeen under the NYU Abu Dhabi Research Institute grant CG008 and ASPIRE Abu Dhabi under Project AARE20-336. The work of M.S. was done within the activities of the Centre for Mathematical Studies, University of Lisbon (CEMS.UL), and was partially supported by the FCT project no. UID/04561/2025. All authors would like to thank for the possibility to participate in the Simons Semester ``Knots, homologies and physics" at the Institute of Mathematics, Polish Academy of Sciences, which was seminal for this work.

\appendix

\section{Stable limit for \texorpdfstring{$(2,2p+1)$}{2,2p+1} torus knots}\label{sec:Stable limit for torus knots}

Let $P_r(T_{2,2j+1})$ be the reduced $Sym^r$-coloured HOMFLY-PT polynomial of the $T_{2,2j+1}$ torus knot. In order to fix the notation and conventions, the value for the $Sym^2$-polynomial for the trefoil - $T_{2,3}$ - is given by:
\begin{equation}\label{p2tre}
P_2(T_{2,3})=a^4 (q^{-4}+q^2+q^4+q^8)-a^6 (1+q^2+q^6+q^8)+a^8q^6.
\end{equation}

\paragraph{Overall shifts with series in $q$}
There are various explicit expression for $P_r(T_{2,2j+1})$. Below we use the one from \cite[formula (3.11)]{FGSS1209} (Note that there the formula is for the homology, so we set $t=-1$, and also replace $q \to q^2$ and $a \to a^2$).
\begin{eqnarray}
P_r(T_{2,2j+1})&=&a^{2jr}q^{-2jr} \sum_{0\le k_j\le k_{j-1}\le\cdots\le k_1\le r} \left[\!
\begin{array}{c}r \\ k_1\end{array}\!\right]\left[\!\begin{array}{c}k_1\\k_2\end{array}\!\right]\cdots\left[\!\begin{array}{c}k_{j-1}\\k_j\end{array}\!\right]\times \label{nfor}\\
&&\quad\quad\quad\times\quad  q^{2(2r+1)(k_1+k_2+\ldots+k_j)-2rk_1-2\sum_{i=1}^{j-1}k_ik_{i+1}}\prod_{i=1}^{k_1}(1-a^2q^{2(i-2)}).\nonumber
\end{eqnarray}
Here we assume unbalanced binomial coefficients:
\begin{eqnarray*}
\left[\!
\begin{array}{c}n\\ k\end{array}\!\right]=\frac{[n]!}{[k]![n-k]!},\quad [n]!=\prod_{i=1}^n [i],\quad [n]=\frac{1-q^{2n}}{1-q^2}.
\end{eqnarray*}
In order to make the series ``start from 1", we have to add overall shift:
\[
\hat{P}_r(T_{2,2j+1})=(a^{-2j}q^{2j})^r P_r(T_{2,2j+1}).
\]

Now the statement is the following:

\begin{proposition}\label{prop1}
\[\lim_{j\to\infty} \hat{P}_r(T_{2,2j+1})=\frac{(1-a^2q^{2r})(1-a^2q^{2r+2})\cdots(1-a^2q^{4r-2})}{(1-q^{2r+2})(1-q^{2r+4})\cdots(1-q^{4r})}.
\]
\end{proposition}

(As usual, $a=q^N$, gives the $SU(N)$ specialization, i.e. expression in products/quotients of the quantum integers.)

First note that in order to pass to the unreduced polynomial, the left-hand side is to be multiplied by the $Sym^r$ polynomial of the unknot, which is (up to an overall factor) equal to
$$\bar{P}_r(unknot)\sim \frac{(1-a^2)(1-a^2q^{2})\cdots(1-a^2q^{2r-2})}{(1-q^{2})(1-q^{4})\cdots(1-q^{2r})}.$$
Multiplying this with right-hand side of the Proposition \ref{prop1} gives exactly
$$\frac{(1-a^2)(1-a^2q^{2})\cdots(1-a^2q^{4r-2})}{(1-q^{2})(1-q^{4})\cdots(1-q^{4r})}\sim \bar{P}_{2r}(unknot).$$

Below we prove directly Proposition \ref{prop1} for $r=2$, but it can be extended for arbitrary $r$ in a~straightforward way. So, our goal is to compute
$$\lim_{j\to\infty} \hat{P}_2(T_{2,2j+1}),$$
where
\begin{eqnarray}
\hat{P}_2(T_{2,2j+1})&=& \sum_{0\le k_j\le k_{j-1}\le\cdots\le k_1\le 2} \left[\!
\begin{array}{c}2 \\ k_1\end{array}\!\right]\left[\!\begin{array}{c}k_1\\k_2\end{array}\!\right]\cdots\left[\!\begin{array}{c}k_{j-1}\\k_j\end{array}\!\right]\times\label{form11}\\
&&\quad\quad\quad\times\quad  q^{10(k_1+k_2+\ldots+k_j)-4k_1-2\sum_{i=1}^{j-1}k_ik_{i+1}}\prod_{i=1}^{k_1}(1-a^2q^{2(i-2)}).\nonumber
\end{eqnarray}
Denote 
\begin{eqnarray*}
\alpha&=&\sharp\{i|k_i=2\},\\
\beta&=&\sharp\{i|k_i=1\}.
\end{eqnarray*}
Then in (\ref{form11}) we can pass to summation over $\alpha$ and $\beta$ with $\alpha+\beta\le j$, while the value of the corresponding summand for different values of $\alpha$ and $\beta$ is given by:
\begin{eqnarray*}
\alpha>0, \,\,\beta>0&:&\quad\quad (1+q^2)q^{12\alpha+8\beta-2}(1-a^2q^{-2})(1-a^2),\\
\alpha>0, \,\,\beta=0&:&\quad\quad q^{12\alpha}(1-a^2q^{-2})(1-a^2), \\
\alpha=0, \,\,\beta>0&:&\quad\quad (1+q^2)q^{8\beta-2}(1-a^2q^{-2}),\\
\alpha=0, \,\,\beta=0&:&\quad\quad 1. \\
\end{eqnarray*}
Therefore
\begin{eqnarray*}
\hat{P}_2(T_{2,2j+1})&=& \sum_{{\begin{array}{c}\alpha+\beta\le j\\
\alpha,\beta>0\end{array}}}(1+q^2)q^{12\alpha+8\beta-2}(1-a^2q^{-2})(1-a^2)+\\
&&+\sum_{0<\alpha\le j}q^{12\alpha}(1-a^2q^{-2})(1-a^2)+\sum_{0<\beta\le j}(1+q^2)q^{8\beta-2}(1-a^2q^{-2})+1,
\end{eqnarray*}
and so
\begin{eqnarray*}
\lim_{j\to\infty}\hat{P}_2(T_{2,2j+1})&=& \sum_{\alpha,\beta>0}(1+q^2)q^{12\alpha+8\beta-2}(1-a^2q^{-2})(1-a^2)+\\
&&+\sum_{\alpha>0}q^{12\alpha}(1-a^2q^{-2})(1-a^2)+\sum_{\beta>0}(1+q^2)q^{8\beta-2}(1-a^2q^{-2})+1=\\
&=&\frac{q^{18}(1+q^2)(1-a^2q^{-2})(1-a^2)}{(1-q^8)(1-q^{12})}+\frac{q^{12}(1-a^2q^{-2})(1-a^2)}{1-q^{12}}+\\
&&+\quad\frac{q^6(1-a^2q^{-2})(1+q^2)}{1-q^8}+1=\\
&=&\frac{(1+q^6)(1-a^2(q^4+q^6)+a^4q^{10})}{(1-q^{8})(1-q^{12})}=\frac{(1-a^2q^{4})(1-a^2q^{6})}{(1-q^{6})(1-q^{8})},
\end{eqnarray*}
as wanted.

\paragraph{Overall shifts with series in $q^{-1}$}
There is also another shift possibility so that expressions for shifted $P_r(T_{2,2j+1})$ start from 1 and are polynomials in $a^2$ and in $q^{-2}$. These are obtained with the following shift:
\[
\hat{P}_r(T_{2,2j+1})=a^{-2jr}q^{-2jr^2} P_r(T_{2,2j+1}).
\]
In particular, for $r=2$:
\begin{equation}
\hat{P}_2(T_{2,2j+1})=(a^{-4}q^{-8})^j P_2(T_{2,2j+1}).
\end{equation}

In order to find the limit when $j\to \infty$, it might be more suitable to use the ``mirror symmetry" expression (see \cite{gorsky2018quadruply}) comparing to (\ref{nfor}):
\begin{eqnarray*}
P_r(T_{2,2j+1})&=&a^{2jr}q^{2jr^2} \sum_{0\le k_j\le k_{j-1}\le\cdots\le k_1\le r} \left[\!
\begin{array}{c}r \\ k_1\end{array}\!\right]_{q^{-1}}\left[\!\begin{array}{c}k_1\\k_2\end{array}\!\right]_{q^{-1}}\cdots\left[\!\begin{array}{c}k_{j-1}\\k_j\end{array}\!\right]_{q^{-1}}\times\\
&&\quad\quad\quad\times\quad  q^{-2(k_1^2+k_2^2+\ldots+k_j^2)-2(k_1+k_2+\ldots+k_j)}\prod_{i=1}^{k_1}(1-a^2q^{2(r-1+i)}).
\end{eqnarray*}
Note that here we continue with the reduced polynomial, and also for the trefoil and $r=2$ we get the exactly same expression (\ref{p2tre}).

Now the statement is the following:

\begin{proposition}\label{prop2}
\[\lim_{j\to\infty} \hat{P}_2(T_{2,2j+1})=\frac{(1-a^2)(1-a^2q^{-2})}{(1-q^{4})(1-q^{6})}.
\]
\end{proposition}

(As usual, $a=q^N$, gives the $SU(N)$ specialization, i.e. expression in products/quotients of the quantum integers.)

Again, we note that in order to pass to the unreduced polynomial, the left-hand side is to be multiplied by the $Sym^r$ polynomial of the unknot, which is (up to an overall factor) equal to
$$\bar{P}_r(unknot)\sim \frac{(1-a^2)(1-a^2q^{2})\cdots(1-a^2q^{2r-2})}{(1-q^{2})(1-q^{4})\cdots(1-q^{2r})}.$$
Multiplying this with right-hand side of the Proposition \ref{prop2} gives exactly
$$\frac{(1-a^2)(1-a^2q^{2})\cdots(1-a^2q^{4r-2})}{(1-q^{2})(1-q^{4})\cdots(1-q^{4r})}\sim \bar{P}_{r \times 2}(unknot),$$
where $r \times 2$ denotes the representation (color) corresponding to the Young tableaux which is a~rectangle with $r$ rows and $2$ columns ($r\times 2$ rectangle).

In order to prove Proposition \ref{prop2} for $r=2$, our goal is to compute
$$\lim_{j\to\infty} \hat{P}_2(T_{2,2j+1}),$$
where
\begin{eqnarray}
\hat{P}_2(T_{2,2j+1})&=& \sum_{0\le k_j\le k_{j-1}\le\cdots\le k_1\le 2} \left[\!
\begin{array}{c}2 \\ k_1\end{array}\!\right]_{q^{-1}}\left[\!\begin{array}{c}k_1\\k_2\end{array}\!\right]_{q^{-1}}\cdots\left[\!\begin{array}{c}k_{j-1}\\k_j\end{array}\!\right]_{q^{-1}}\times\label{form22}\\
&&\quad\quad\quad\times\quad  q^{-2(k_1^2+k_2^2+\ldots+k_j^2)-2(k_1+k_2+\ldots+k_j)}\prod_{i=1}^{k_1}(1-a^2q^{2(1+i)}).\nonumber
\end{eqnarray}
Again, let us denote 
\begin{eqnarray*}
\alpha&=&\sharp\{i|k_i=2\},\\
\beta&=&\sharp\{i|k_i=1\}.
\end{eqnarray*}
Then in (\ref{form22}) we can pass to summation over $\alpha$ and $\beta$ with $\alpha+\beta\le j$, while the value of the corresponding summand for different values of $\alpha$ and $\beta$ is given by:
\begin{eqnarray*}
\alpha>0, \,\,\beta>0&:&\quad\quad (1+q^2)q^{12\alpha+8\beta-2}(1-a^2q^{-2})(1-a^2),\\
\alpha>0, \,\,\beta=0&:&\quad\quad q^{12\alpha}(1-a^2q^{-2})(1-a^2), \\
\alpha=0, \,\,\beta>0&:&\quad\quad (1+q^2)q^{8\beta-2}(1-a^2q^{-2}),\\
\alpha=0, \,\,\beta=0&:&\quad\quad 1. \\
\end{eqnarray*}
Therefore
\begin{eqnarray*}
\hat{P}_2(T_{2,2j+1})&=& \sum_{{\begin{array}{c}\alpha+\beta\le j\\
\alpha,\beta>0\end{array}}}(1+q^2)q^{12\alpha+8\beta-2}(1-a^2q^{-2})(1-a^2)+\\
&&+\sum_{0<\alpha\le j}q^{12\alpha}(1-a^2q^{-2})(1-a^2)+\sum_{0<\beta\le j}(1+q^2)q^{8\beta-2}(1-a^2q^{-2})+1,
\end{eqnarray*}
and so
\begin{eqnarray*}
\lim_{j\to\infty}\hat{P}_2(T_{2,2j+1})&=& \sum_{\alpha,\beta>0}(1+q^2)q^{12\alpha+8\beta-2}(1-a^2q^{-2})(1-a^2)+\\
&&+\sum_{\alpha>0}q^{12\alpha}(1-a^2q^{-2})(1-a^2)+\sum_{\beta>0}(1+q^2)q^{8\beta-2}(1-a^2q^{-2})+1=\\
&=&\frac{q^{18}(1+q^2)(1-a^2q^{-2})(1-a^2)}{(1-q^8)(1-q^{12})}+\frac{q^{12}(1-a^2q^{-2})(1-a^2)}{1-q^{12}}+\\
&&+\quad\frac{q^6(1-a^2q^{-2})(1+q^2)}{1-q^8}+1=\\
&=&\frac{(1+q^6)(1-a^2(q^4+q^6)+a^4q^{10})}{(1-q^{8})(1-q^{12})}=\frac{(1-a^2q^{4})(1-a^2q^{6})}{(1-q^{6})(1-q^{8})},
\end{eqnarray*}
as wanted.

\section{Splitting by linking}\label{app:splitting_by_linking}

Here we prove Proposition \ref{prp:splitting_by_linking}, analogously to Proposition 6.1 from \cite{KLNS2312}.

\begin{proof}
Let us consider  $(k,l)$-splitting of $n$ nodes of quiver $Q$ with permutation $\sigma\in S_{n}$, in the presence of $m-n$ spectators with corresponding shifts $h_{1},\dots,h_{m-n}$, and multiplicative factor $\kappa$ (see Definition \ref{def:splitting}). 

Let us consider an~augmented quiver $Q^{+}$ given by $Q$ plus one extra node. We will denote its index by $\iota$ and take $x_{\iota}= \kappa$ (below we always assume that this is the last row/column; recall that we write only the upper-triangular part of the symmetric matrix):
\begin{equation}
\label{eq:checkQ' definition}
\begin{split}
    C^{+}&=\left[\begin{array}{cccccc|c}
     &  &  &  &  &  & h_{1}\\
     &  &  &  &  &  & \vdots\\
     &  &  &  &  &  & h_{m-n}\\
     &  &  & C &  &  & k\\
     &  &  &  &  &  & \vdots\\
     &  &  &  &  &  & k\\
     \hline
    &  &  &  &  &  & l-2k
    \end{array}\right]    \\
    &=\left[\begin{array}{cccccc|c}
    C_{11} &  &  &  & \cdots &  & h_{1}\\
     & \ddots &  &  &  &  & \vdots\\
     &  & C_{m-n,m-n} &  & \cdots &  & h_{m-n}\\
     &  &  & C_{m-n+1,m-n+1} & \cdots &  & k\\
     &  &  &  & \ddots &  & \vdots\\
     &  &  &  &  & C_{mm} & k\\
     \hline
    &  &  &  &  &  & l-2k
    \end{array}\right].  
\end{split}
\end{equation}
Without loss of generality, we consider linkings $L(i\iota)$ and $L(j\iota)$ ($i,j\in\{m-n+1,\dots,m\},i\neq j$) acting in two different orders (we allow for other linkings from the set $\{L(m-n+1,\iota),\dots,L(m,\iota)\}$ to act before, in between, or after them). On one hand, we have
\begin{align}
 \dots  L(j\iota) &\dots  L(i\iota)\dots C^{+}= \nonumber \\
&=  \dots L(j\iota)\dots L(i\iota)\dots
\renewcommand\arraystretch{1.2}
\setlength\arraycolsep{0.5pt}
\left[\begin{array}{ccccccc|c}
\ddots &  &  & \vdots &  & \vdots &  & \vdots\\
 & C_{m-n,m-n} & \cdots & C_{m-n,i} & \cdots & C_{m-n,j} & \cdots & h_{m-n}\\
&  &  \ddots & \vdots &  & \vdots &  & \vdots\\
&  &  &  C_{ii} & \cdots & C_{ij} & \cdots & k\\
&  &  &  & \ddots & \vdots &  & \vdots\\
&  &  &  &  & C_{jj} & \cdots & k\\
&  &  &  &  &  & \ddots & \vdots\\
\hline
&  &  &  &  &  &  & l-2k
\end{array}\right] \label{eq:splitting from linking 1}\\
&=  \dots L(j\iota)\dots
\renewcommand\arraystretch{1.2}
\setlength\arraycolsep{0.5pt}
\left[\begin{array}{cccccccc|c}
\ddots &  &  & \vdots & \vdots &  & \vdots &  & \vdots\\
 & C_{m-n,m-n} & \cdots & C_{m-n,i} & C_{m-n,i}+h_{m-n} & \cdots & C_{m-n,j} & \cdots & h_{m-n}\\
& & \ddots & \vdots &  &  & \vdots &  & \vdots\\
& & & C_{ii} & C_{ii}+k & \cdots & C_{ij} & \cdots & k+1\\
& & & & C_{ii}+l & \cdots & C_{ij}+k & \cdots & l-k\\
& & & & & \ddots & \vdots &  & \vdots\\
& & & & & & C_{jj} & \cdots & k\\
& & & & & & & \ddots & \vdots\\
\hline
& & & & & & & & l-2k
\end{array}\right] \nonumber
\end{align}
\begin{align}
&~~~~~~~~=
\renewcommand\arraystretch{1.2}
\setlength\arraycolsep{0.5pt}
\left[\begin{array}{ccccccccc|c}
\ddots &  &  & \vdots & \vdots &  & \vdots & \vdots &  & \vdots\\
 & C_{m-n,m-n} & \cdots & C_{m-n,i} & C_{m-n,i}+h_{m-n} & \cdots & C_{m-n,j} & C_{m-n,j}+h_{m-n} & \cdots & h_{m-n}\\
&  & \ddots & \vdots & \vdots &  & \vdots & \vdots &  & \vdots\\
&  &  & C_{ii} & C_{ii}+k & \cdots & C_{ij} & C_{ij}+k+1 & \cdots & k+1\\
&  &  &  & C_{ii}+l & \cdots & C_{ij}+k & C_{ij}+l & \cdots & l-k\\
&  &  &  &  & \ddots & \vdots & \vdots &  & \vdots\\
&  &  &  &  &  & C_{jj} & C_{jj}+k & \cdots & k+1\\
&  &  &  &  &  &  & C_{jj}+l & \cdots & l-k\\
 &  &  &  &  &  &  &  & \ddots & \vdots\\
 \hline
 &  &  &  &  &  &  &  &  & l-2k
\end{array}\right]\,, \nonumber
\end{align}
on the other:
\begin{align}  
\label{eq:splitting from linking 2}
\dots  L(i\iota) & \dots  L(j\iota)\dots C^{+}=\\
= &
\renewcommand\arraystretch{1.5}
\setlength\arraycolsep{0.5pt}
\left[\begin{array}{ccccccccc|c}
\ddots &  &  & \vdots & \vdots &  & \vdots & \vdots &  & \vdots\\
 & C_{m-n,m-n} & \cdots & C_{m-n,i} & C_{m-n,i}+h_{m-n} & \cdots & C_{m-n,j} & C_{m-n,j}+h_{m-n} & \cdots & h_{m-n}\\
&  & \ddots & \vdots & \vdots &  & \vdots & \vdots &  & \vdots\\
&  &  & C_{ii} & C_{ii}+k & \cdots & C_{ij} & C_{ij}+k & \cdots & k+1\\
&  &  &  & C_{ii}+l & \cdots & C_{ij}+k+1 & C_{ij}+l & \cdots & l-k\\
&  &  &  &  & \ddots & \vdots & \vdots &  & \vdots\\
&  &  &  &  &  & C_{jj} & C_{jj}+k & \cdots & k+1\\
&  &  &  &  &  &  & C_{jj}+l & \cdots & l-k\\
 &  &  &  &  &  &  &  & \ddots & \vdots\\
 \hline
 &  &  &  &  &  &  &  &  & l-2k
\end{array}\right]. \nonumber
\end{align}
In both cases the variables associated to new nodes arising from linkings
$L(i\iota)$ and $L(j\iota)$ are given by $x_{i}\kappa$ and $x_{j}\kappa$.

Analysing these adjacency matrices and referring to Definition \ref{def:splitting}, we can see that $L(m-n+\sigma(1),\iota)\dots L(m-n+\sigma(n),\iota)Q^{+}$ is the matrix that arises from $(k,l)$-splitting of $n$ nodes of $Q$ with permutation $\sigma\in S_{n}$ in the presence of $m-n$ spectators with corresponding shifts $h_{1},\dots,h_{m-n}$, multiplicative factor $\kappa$, and one extra row/column 
\begin{equation}
    (h_{1}\,,\,\dots,h_{s}\,,\,\dots\,, \,k+1\,,\,l-k\,,\,\dots\, ,\,k+1,\,l-k\,,\,\dots\,\vert \,l-2k)\,,
\end{equation}
that we denoted by $\iota$.\footnote{One can compare matrices (\ref{eq:splitting from linking 1}) and (\ref{eq:splitting from linking 2}) with  Definition \ref{def:splitting} and check that the position of the extra unit follows exactly the inverses in permutation $\sigma$.}
\end{proof}

\bibliography{refs}

\bibliographystyle{JHEP}

\end{document}